\documentclass[11pt,letterpaper]{article}

\usepackage[T1]{fontenc}
\usepackage[letterpaper,margin=1in]{geometry}
\usepackage{amsmath,amssymb,amsthm,mathtools}
\usepackage{aliascnt}
\usepackage{newpxtext,newpxmath}
\usepackage{microtype}
\usepackage{enumitem}
\usepackage{xcolor}
\usepackage{tikz}
\usetikzlibrary{arrows.meta}
\usepackage{pgfplots}
\pgfplotsset{compat=1.17}
\usepackage{booktabs}
\usepackage[colorlinks=true,
            linkcolor=blue!55!black,
            citecolor=blue!55!black,
            urlcolor=blue!55!black]{hyperref}
\usepackage[nameinlink,capitalise,noabbrev]{cleveref}

\allowdisplaybreaks
\setlist[itemize]{leftmargin=2em,itemsep=2pt,topsep=4pt}
\setlist[enumerate]{leftmargin=2.4em,itemsep=2pt,topsep=4pt}

\newtheorem{theorem}{Theorem}[section]

\newaliascnt{lemma}{theorem}
\newtheorem{lemma}[lemma]{Lemma}
\aliascntresetthe{lemma}

\newaliascnt{proposition}{theorem}
\newtheorem{proposition}[proposition]{Proposition}
\aliascntresetthe{proposition}

\newaliascnt{corollary}{theorem}
\newtheorem{corollary}[corollary]{Corollary}
\aliascntresetthe{corollary}

\theoremstyle{definition}

\newaliascnt{definition}{theorem}
\newtheorem{definition}[definition]{Definition}
\aliascntresetthe{definition}

\theoremstyle{remark}

\newaliascnt{remark}{theorem}
\newtheorem{remark}[remark]{Remark}
\aliascntresetthe{remark}

\crefname{theorem}{Theorem}{Theorems}
\Crefname{theorem}{Theorem}{Theorems}

\crefname{lemma}{Lemma}{Lemmas}
\Crefname{lemma}{Lemma}{Lemmas}

\crefname{proposition}{Proposition}{Propositions}
\Crefname{proposition}{Proposition}{Propositions}

\crefname{corollary}{Corollary}{Corollaries}
\Crefname{corollary}{Corollary}{Corollaries}

\crefname{definition}{Definition}{Definitions}
\Crefname{definition}{Definition}{Definitions}

\crefname{remark}{Remark}{Remarks}
\Crefname{remark}{Remark}{Remarks}

\AddToHook{env/theorem/begin}{\crefalias{section}{theorem}}
\AddToHook{env/lemma/begin}{\crefalias{section}{lemma}}
\AddToHook{env/proposition/begin}{\crefalias{section}{proposition}}
\AddToHook{env/corollary/begin}{\crefalias{section}{corollary}}
\AddToHook{env/definition/begin}{\crefalias{section}{definition}}
\AddToHook{env/remark/begin}{\crefalias{section}{remark}}

\newcommand{\E}{\mathbb{E}}
\newcommand{\R}{\mathbb{R}}
\newcommand{\1}{\mathbf 1}
\newcommand{\cF}{\mathcal{F}}
\newcommand{\cG}{\mathcal{G}}
\newcommand{\cP}{\mathcal{P}}

\newcommand{\cR}{\mathcal{R}}
\newcommand{\cB}{\mathcal{B}}
\newcommand{\FB}{\operatorname{FB}}
\newcommand{\SB}{\operatorname{SB}}
\newcommand{\GFT}{\operatorname{GFT}}

\newcommand{\supp}{\operatorname{supp}}
\newcommand{\Beta}{\operatorname{Beta}}
\newcommand{\Lag}{\mathcal{L}}
\newcommand{\dd}{\,\mathrm d}
\newcommand{\vct}[1]{\boldsymbol{#1}}

\title{Second-Best Gains from Trade in Matching Markets}
\author{
Xiaohui Bei\thanks{School of Physical and Mathematical Sciences, Nanyang Technological University, Singapore.\protect\\
Email: \href{mailto:xhbei@ntu.edu.sg}{\nolinkurl{xhbei@ntu.edu.sg}}.}
\and
Bo Li\thanks{Department of Computing, The Hong Kong Polytechnic University, Hong Kong.\protect\\
Email: \href{mailto:comp-bo.li@polyu.edu.hk}{\nolinkurl{comp-bo.li@polyu.edu.hk}}.}
\and
Wenhao Wu\thanks{School of Physical and Mathematical Sciences, Nanyang Technological University, Singapore.\protect\\
Email: \href{mailto:wenhao008@e.ntu.edu.sg}{\nolinkurl{wenhao008@e.ntu.edu.sg}}.}
\and
Shengwei Zhou\thanks{School of Physical and Mathematical Sciences, Nanyang Technological University, Singapore.\protect\\
Email: \href{mailto:s.arthur.zhou@gmail.com}{\nolinkurl{s.arthur.zhou@gmail.com}}.}
}
\date{}

\begin{document}

\maketitle

\begin{abstract}
\normalsize
We study gains from trade (GFT) in two-sided matching markets with independent private types and arbitrary downward-closed feasibility constraints.
The second-best benchmark is the maximum expected GFT achievable by a Bayesian incentive-compatible, interim individually rational mechanism that is strongly budget balanced at every report profile.
These constraints generally preclude attaining the first-best GFT and raise the question of how much efficiency must be lost.
We prove that the second-best GFT is at least one half of the first-best GFT in every such matching market.
This recovers and generalizes the recent $1/2$ guarantee for bilateral trade by Liu, Qin, Ren, and Wang~\cite{LiuQinRenWang2026} to markets with multiple buyers and sellers and arbitrary downward-closed feasibility constraints.
Together with their matching upper bound for bilateral trade, our result establishes a tight worst-case ratio of $1/2$ for this general class of matching markets.
Our proof builds on the virtual-GFT framework of Br\"ustle, Cai, Wu, and Zhao~\cite{BrustleCaiWuZhao2017} to reduce the problem to a one-parameter Lagrangian.
The main step is a geometric, edge-by-edge analysis based on first-best edge-selection regions, combined with a randomized contraction in rank space.
The same analysis also determines the exact worst-case efficiency for every level of outside subsidy. In particular, even a small outside subsidy can significantly improve the worst-case guarantee.
\end{abstract}

\clearpage
\tableofcontents
\clearpage

\section{Introduction}

Two-sided markets play a central role in modern economies, bringing together buyers and sellers in settings ranging from labor and financial markets to online platforms for advertising, ride-sharing, and accommodation.
A central goal in designing such markets is economic efficiency: selecting a feasible set of trades that maximizes the total welfare they generate.
The resulting increase in welfare is known as the \emph{gains from trade} (GFT).
If buyers' values and sellers' costs were publicly known, the market designer would simply choose a feasible matching that maximizes the sum of value-minus-cost surpluses.
This full-information benchmark is the \emph{first best}.
The difficulty is that buyers privately know their values and sellers privately know their costs, so both sides may report strategically to increase their own utility.
This tension makes the first best generally unattainable.
The Myerson--Satterthwaite theorem already shows that, even in bilateral trade with a single buyer and a single seller, ex-post efficiency is incompatible with Bayesian incentive compatibility, individual rationality, and budget balance under standard nondegeneracy conditions~\cite{MyersonSatterthwaite1983}.
This motivates the \emph{second best}: the mechanism that maximizes expected GFT subject to the desired incentive, participation, and budget constraints.
The quantitative question we study is therefore
\begin{quote}
    \emph{What is the worst-case ratio of second-best GFT to first-best GFT?}
\end{quote}

A sequence of works has studied approximations to the first-best GFT in the simpler bilateral trade setting.
Deng, Mao, Sivan, and Wang established the first constant approximation~\cite{DengMaoSivanWang2022}, and Fei subsequently improved the guarantee to $1/3.15$~\cite{Fei2022}.
Recent work has established that the exact worst-case efficiency ratio of the bilateral second best is one half~\cite{LiuQinRenWang2026}.

Matching markets generalize bilateral trade to multiple buyers and sellers connected by a bipartite graph, with trades allowed only along edges.
Existing work has established strong guarantees for simple mechanisms, specifically, the random-offerer mechanism, which randomizes uniformly between the generalized seller-offering and buyer-offering mechanisms.
Such a mechanism is dominant-strategy incentive compatible (DSIC), ex-post individually rational (IR), and ex-ante weakly budget balanced.
Br\"ustle, Cai, Wu, and Zhao developed these generalized offering mechanisms and showed that the better of the two achieves one half of the second-best GFT~\cite{BrustleCaiWuZhao2017}.
More recently, Babaioff, Rubinstein, Tan, and Wang proved that this equal randomization achieves a $1/3.15$-approximation to the \emph{first-best} GFT in matching markets~\cite{BabaioffRubinsteinTanWang2026}.
These results give strong guarantees for simple mechanisms, but they do not identify the exact gap between the unrestricted second best and the first best.
Together with the bilateral tight examples of Liu, Qin, Ren, and Wang~\cite{LiuQinRenWang2026}, they leave the universal worst-case ratio for this matching-market class between $1/3.15$ and $1/2$.

This paper closes that gap: every instance in this class satisfies the same $1/2$ guarantee, and this universal constant is tight over the class already in the bilateral special case.

\subsection{Our results}\label{subsec:our-results}

We consider a matching market with finitely many unit-demand buyers and unit-supply sellers.
Types are mutually independent and supported on $[0,1]$.
A feasible outcome is a matching from a fixed downward-closed family $\cF$.
Let $\FB$ denote the expected GFT of the first-best matching, which is the maximum welfare outcome at each profile, and let $\SB$ denote the maximum GFT attainable by a mechanism that is Bayesian incentive compatible (BIC), interim individually rational (interim IR), and exactly strongly budget balanced at every report profile.

Our main theorem determines the exact efficiency loss.

\begin{theorem}[informal]
For every matching-market instance $\mathcal I$,
\begin{equation*}
  \SB(\mathcal I)\ge \frac12\FB(\mathcal I).
\end{equation*}
Moreover, this ratio is tight.
\end{theorem}

The lower bound is new for general downward-closed matching markets.
Tightness follows already from bilateral trade: bilateral trade is a special case of our model, and~\cite{LiuQinRenWang2026} construct instances whose second-best ratio approaches $1/2$.

We also apply this new lower bound result to multi-dimensional markets.
By utilizing the reduction from~\cite{CaiGoldnerMaZhao2021}, one can show that their mechanism achieves a $1/4$ approximation to the first-best GFT for multi-dimensional matching markets with one unit-demand buyer and multiple unit-supply sellers, improving the previous $1/6.3$ guarantee established in~\cite{BabaioffRubinsteinTanWang2026}.

\begin{corollary}[informal]
  For every multi-dimensional matching-market instance $\mathcal I$ with multiple unit-supply sellers and one unit-demand buyer whose item values are mutually independent, the mechanism of~\cite{CaiGoldnerMaZhao2021} achieves at least a $1/4$ fraction of the first-best GFT.
\end{corollary}

\paragraph{Subsidized markets.}
We also consider markets in which an outside party may cover part of the mechanism's deficit.
If the expected deficit may be as large as $\delta\FB$, then the best efficiency guarantee that holds for all instances is
\begin{equation*}
  \gamma(\delta)=\inf_{\alpha\ge0}\bigl\{g_\alpha+\alpha\delta\bigr\},
  \qquad
  g_\alpha=\frac{1+\alpha}{\mathrm B\bigl(\frac{1}{1+\alpha},\frac{1}{1+\alpha}\bigr)},
\end{equation*}
where $\mathrm B$ is the Beta function and $g_\alpha$ is the coefficient that our analysis produces for the Lagrange multiplier~$\alpha$ (\cref{sec:subsidy}).
In particular, $\gamma(0)=1/2$, the guarantee grows like $\gamma(\delta)=\frac12+\Theta(\delta^{2/3})$ for small $\delta$, and $\gamma(\delta)=1$ if and only if $\delta\ge1$.
This result implies that even a small outside subsidy can significantly improve the worst-case efficiency guarantee.
For example, a subsidy of $1\%$ of the first-best GFT raises the worst-case guarantee to at least $56.6\%$, while a subsidy of $10\%$ raises it to around $73.5\%$.
The proof also yields an alternative, fully discrete family of tight instances for the half guarantee.

\subsection{Technical overview}\label{subsec:tech-overview}
To prove the 1/2 lower bound, we adopt the virtual-GFT characterization of~\cite{BrustleCaiWuZhao2017} with strong LP duality, which gives the following Lagrangian formulation:
\begin{equation*}
  \SB=\inf_{\alpha\ge0}\max_{q\in\cP}
    \bigl\{\cG(q)+\alpha\cB(q)\bigr\}.
\end{equation*}
Here $\cP$ is the polytope of allocation rules that select only feasible matchings at every report profile and have monotone interim service, while $\cG(q)$ and $\cB(q)$ denote ordinary and virtual GFT, respectively.
Our lower-bound argument starts from this identity: for every fixed budget price $\alpha$, it suffices to produce a monotone feasible allocation with Lagrangian value at least $\FB/2$.
The argument proceeds in two steps.
\begin{itemize}[leftmargin=*]
\item \textbf{An edge-level Lagrangian bound via rank-space contraction.}
Fix an edge $e=(i,j)$ and all reports other than $(v_i,c_j)$.
We characterize the buyer values and seller costs for which $e$ belongs to the tie-broken first-best matching.
Using the buyer's upper-tail rank $\rho$ and the seller's lower-tail rank $\sigma$ under the \emph{original} marginal distributions, this edge-selection region becomes a coordinatewise down-set $\mathcal D\subseteq(0,1)^2$.
Its first-best contribution is
\begin{equation*}
  \int_{\mathcal D}\bigl(Q_i(\rho)-Q_j(\sigma)\bigr)\,\dd\rho\,\dd\sigma.
\end{equation*}
Here, $Q_i$ is buyer $i$'s upper-tail quantile function and $Q_j$ is seller $j$'s lower-tail quantile function.
For each $\alpha>0$, we choose scale factors $\theta$ and $\eta$ as functions of a common random variable $Z$, with $\theta$ increasing and $\eta$ decreasing in $Z$.
We then contract $\mathcal D$ coordinatewise to
\begin{equation*}
  \mathcal D_Z
  :=
  \{(\theta\rho,\eta\sigma):(\rho,\sigma)\in\mathcal D\}.
\end{equation*}
The contraction satisfies $\mathcal D_Z\subseteq\mathcal D$.
The contracted region is used to lower-bound the fixed-$\alpha$ Lagrangian objective.
Bounding the buyer and seller terms separately shows that its expected Lagrangian contribution satisfies
\begin{equation*}
  \E_Z\bigl[\text{local Lagrangian on }\mathcal D_Z\bigr]
  \ge
  g_\alpha\cdot\text{local first-best surplus},
  \qquad
  g_\alpha>\frac12.
\end{equation*}
Here a reported type with positive probability is handled by drawing a hidden uniform rank from its quantile interval.

\item \textbf{From edge-level bounds to a monotone feasible allocation.}
Next, we discuss how to assemble the edge-level bounds into a monotone feasible allocation.
For each profile, apply the local construction to every edge of the tie-broken first-best matching and retain each such edge with its induced probability.
Note that every realized outcome is a subset of one feasible first-best matching and is therefore feasible by downward closedness.
For every $\alpha>0$, summing the edge-level inequalities over external reports and edges yields
\begin{equation*}
  \cG(q^\alpha)+\alpha\cB(q^\alpha)
  \ge\frac12\FB.
\end{equation*}
At $\alpha=0$, the first-best allocation itself is monotone.
Taking the infimum in the Lagrangian identity proves the half guarantee.
\end{itemize}
Finally, the allocations $q^\alpha$ need not satisfy $\cB(q^\alpha)\ge0$ for a fixed $\alpha>0$.
However, $\cB(q^\alpha)$ is continuous in $\alpha$ and positive for sufficiently large $\alpha$, so the family $(q^\alpha)_{\alpha\ge0}$ itself contains an allocation with $\cB\ge0$ and GFT at least $\FB/2$, which can be implemented as a BIC, interim-IR, and exactly strongly budget-balanced mechanism (\cref{subsec:assembly}).

In \cref{sec:subsidy}, we show that the coefficient $g_\alpha$ cannot be improved for any fixed $\alpha$, using bilateral instances in which $v+\alpha\phi(v)$ is constant below the highest value.
Since an outside subsidy $s$ only adds $\alpha s$ to the Lagrangian, this determines the best efficiency guarantee for every subsidy level.

\subsection{Related work}\label{subsec:related-work}

\paragraph{Bilateral trade.}
The foundational impossibility theorem is due to Myerson and Satterthwaite~\cite{MyersonSatterthwaite1983}, while Myerson's optimal-auction analysis developed the payment and virtual-value tools that underlie Bayesian mechanism design~\cite{Myerson1981}.
A sequence of works directly approximated first-best GFT using Bayesian mechanisms.
\cite{DengMaoSivanWang2022} gave the first distribution-free constant guarantee for the random-offerer mechanism, proving that it obtains at least a $1/8.23$ fraction of first-best GFT.
Fei improved this guarantee to approximately $1/3.15$~\cite{Fei2022}, and Hartline and Wang later gave a geometric analysis recovering the same factor~\cite{HartlineWang2025}.
On the negative side, Babaioff, Dobzinski, and Kupfer~\cite{BabaioffDobzinskiKupfer2021} refuted a natural conjecture that the random-offerer mechanism always obtains one half of first-best GFT, and Cai, Gupta, Li, and Mehta~\cite{CaiGuptaLiMehta2026} subsequently found a stronger counterexample.
Jo recently improved the random-offerer lower guarantee to $1/\pi$ and constructed instances on which its first-best fraction is below $0.460243$~\cite{Jo2026}.
For the unrestricted second-best benchmark, Blumrosen and Mizrahi~\cite{BlumrosenMizrahi2016} established the previous $2/e$ upper bound, while~\cite{LiuQinRenWang2026} later determined the exact worst-case efficiency to be $1/2$.
A separate line studies bilateral trade under stronger incentive requirements through fixed-price mechanisms~\cite{HagertyRogerson1987,McAfee2008,BlumrosenDobzinski2021,KangPerniceVondrak2022,LiuRenWang2023,CaiWu2023,GiambartolomeiDeKeijzer2026}.
Other extensions study bilateral trade with samples~\cite{DengMaoSivanWangWu2025} or a strategic broker~\cite{HajiaghayiHajiaghayiPengShin2025}, as well as repeated trade under online-learning models~\cite{CesaBianchiCesariColomboniFuscoLeonardi2024,BernasconiCastiglioniCelliFusco2024}.

\paragraph{Two-sided and matching markets.}
Beyond the single-buyer, single-seller setting, a broad literature studies double auctions and two-sided markets with multiple buyers and sellers.
For double auctions, McAfee introduced a dominant-strategy mechanism that may discard the least valuable efficient trade in order to maintain budget balance~\cite{McAfee1992}.
Subsequent work developed constant-factor welfare approximations and modular mechanisms for double auctions and richer two-sided combinatorial settings~\cite{ColiniBaldeschiDeKeijzerLeonardiTurchetta2016,DuttingRoughgardenTalgamCohen2014,ColiniBaldeschiGoldbergDeKeijzerLeonardiRoughgardenTurchetta2020}.
\cite{BrustleCaiWuZhao2017} developed the two-sided duality and generalized offering mechanisms that underlie much of the subsequent literature.
Their mechanisms accommodate arbitrary downward-closed feasibility constraints, and the better of the generalized seller- and buyer-offering mechanisms achieves one half of the optimal BIC/interim-IR ex-ante weakly budget-balanced GFT.
They also give an expected-to-strong budget-balancing transformation and characterize implementability in terms of monotonicity and virtual GFT~\cite{BrustleCaiWuZhao2017}.
Babaioff, Cai, Gonczarowski, and Zhao constructed mechanisms for double auctions and matching markets that combine a constant ex-ante approximation to second-best GFT with an ex-post guarantee that becomes asymptotically first-best in large markets~\cite{BabaioffCaiGonczarowskiZhao2018}.
Cai, Goldner, Ma, and Zhao studied multi-dimensional two-sided markets and highlighted the first-best/second-best gap in matching-type environments as an important obstacle~\cite{CaiGoldnerMaZhao2021}.
\cite{BabaioffRubinsteinTanWang2026} subsequently obtained a $1/3.15$-approximation to the first-best for single-dimensional matching markets with arbitrary downward-closed feasible matching families by randomizing between generalized seller-offering and buyer-offering mechanisms.
Their mechanism is DSIC, ex-post IR, and ex-ante weakly budget balanced.

\paragraph{Relation to recent bilateral results.}
Our tight constant necessarily matches the bilateral result of~\cite{LiuQinRenWang2026} because bilateral trade is contained in our model.
The two proofs take different routes after their respective optimization/Lagrangian formulations.
In~\cite{LiuQinRenWang2026}, the lower bound is converted into a greedy allocation problem and analyzed through a state representation and dominance order.
Our matching-market proof instead starts from the tie-broken first-best matching and, after fixing external reports, characterizes the two-dimensional region in which a given edge is selected.
It then contracts the associated rank-space down-set using two coupled random scale factors, one increasing as the other decreases, before assembling the local bounds across edges.

\paragraph{Organization.}
Section~2 introduces the model and records the expected-to-profile-wise balancing lemma.
Section~3 states the finite-support reduction and develops the finite-support virtual-GFT/Lagrangian formulation.
Section~4 proves the tight half-efficiency lower bound for finite-support priors via the geometric contraction and global assembly, followed by applications to multi-dimensional markets.
Section~5 studies markets with an outside subsidy.
Section~6 concludes the paper.
The appendices contain technical details and deferred proofs.

\section{Preliminaries}
\label{sec:prelim}
\subsection{Model and setup}
We consider a general two-sided matching market with a set of unit-demand buyers $B$ and a set of unit-supply sellers $S$, and $K = B \cup S$.
Each buyer $i\in B$ has a value $v_i \in [0,1]$ for receiving an item, while each seller $j\in S$ has a private cost $c_j \in [0,1]$ for providing one.
The value $v_i$ of each buyer $i\in B$ is independently drawn from a distribution $F_i$, which is known in advance.
Similarly, the cost $c_j$ of each seller $j$ is also independently drawn from a known distribution $G_j$.
Throughout, a probability distribution on $[0,1]$ is understood as a Borel probability measure.
A realized profile is written $\vct t = (\vct v, \vct c) \in [0,1]^K$.
Let $\pi$ denote the product distribution induced by the mutually independent priors $(F_i)_{i\in B}$ and $(G_j)_{j\in S}$.
We use $v_i$ and $c_j$ both for type coordinates and for their realizations.
In probabilistic expressions, they are understood to be drawn from $F_i$ and $G_j$, respectively; the intended meaning will always be clear from context.

Let $E\subseteq B\times S$ be the set of potential trading edges.
A set $M\subseteq E$ is a matching if no agent is incident to more than one edge of $M$.
The market is equipped with a nonempty finite family $\cF\subseteq 2^E$ satisfying:
\begin{enumerate}[label=\textup{(\roman*)}]
  \item every $M\in\cF$ is a matching;
  \item $\cF$ is downward closed: if $M\in\cF$ and $M'\subseteq M$, then $M'\in\cF$.
\end{enumerate}
In particular, $\varnothing\in\cF$.
The formulation contains ordinary graph matching as well as additional downward-closed restrictions, such as a cap on the total number of trades or restrictions that exclude specified combinations of otherwise compatible edges.

We refer to
\begin{equation*}
\mathcal{I}
:=
\bigl(
B,S,E,\mathcal{F},
(F_i)_{i\in B},
(G_j)_{j\in S}
\bigr)
\end{equation*}
as an instance of the market.

We consider randomized mechanisms: a randomized allocation rule is a Borel kernel
\begin{equation*}
  q:[0,1]^K\longrightarrow\Delta(\cF).
\end{equation*}
For every report profile $\vct t$, write $q(\vct t)=(q_M(\vct t))_{M\in\cF}$, where $q_M(\vct t)$ is the probability of selecting $M$.
Its edge marginals are
\begin{equation*}
  q_{ij}(\vct t) :=\sum_{\substack{M\in\cF\\(i,j)\in M}}q_M(\vct t).
\end{equation*}
For each buyer $i$, each seller $j$, the total service probabilities at profile $\vct t$ are
\begin{equation*}
  q_i(\vct t)=\sum_{j':(i,j')\in E}q_{ij'}(\vct t), \qquad q_j(\vct t)=\sum_{i':(i',j)\in E}q_{i'j}(\vct t).
\end{equation*}
Because every feasible outcome is a matching, both quantities lie in $[0,1]$.

\paragraph{Gains from Trade.}
The expected GFT of $q$ conditional on the profile $\vct t$ is
\begin{equation*}
  \GFT(q;\vct t) =\sum_{(i,j)\in E}(v_i-c_j)q_{ij}(\vct t),
\end{equation*}
and the ex-ante GFT is
\begin{equation*}
  \GFT(q)=\E_{\vct t\sim\pi}[\GFT(q;\vct t)].
\end{equation*}

\subsection{Incentives, rationality, and balancing}

A direct mechanism consists of an allocation rule $q$ and a transfer rule.
For each report profile $\vct t$, let $p_i(\vct t)\in\R$ denote buyer $i$'s expected payment to the mechanism and let $r_j(\vct t)\in\R$ denote seller $j$'s expected receipt, where both expectations are conditional on $\vct t$.
These amounts may be negative because we do not require an ex-post guarantee.
We next define interim quantities by fixing an agent's own type or report and averaging over the other agents' types and the mechanism's randomization.

As is standard in the literature, we assume throughout that all transfer expectations appearing below are finite and well defined; see, e.g.,~\cite{BorgersNorman2009}.

Let $\vct t_{-i}$ (resp. $\vct t_{-j}$) be the profile of $K \setminus \{i\}$ (resp. $K \setminus \{j\}$) when fixing a buyer $i$ (resp. a seller $j$).
Then we can define buyer and seller interim service rules by
\begin{align*}
  x_i(v) &=\E_{\vct t_{-i}} \bigl[q_i(v,\vct t_{-i})\bigr],\\
  y_j(c) &=\E_{\vct t_{-j}} \bigl[q_j(c,\vct t_{-j})\bigr].
\end{align*}
Their interim transfers are
\begin{equation*}
  p_i(v)=\E_{\vct t_{-i}}[p_i(v,\vct t_{-i})], \qquad r_j(c)=\E_{\vct t_{-j}}[r_j(c,\vct t_{-j})].
\end{equation*}
If buyer $i$ has true value $v$ and reports $z$, her interim expected utility is
\begin{equation*}
  u_i(v;z)=v x_i(z)-p_i(z).
\end{equation*}
If seller $j$ has true cost $c$ and reports $z$, her interim expected utility is
\begin{equation*}
  u_j(c;z)=r_j(z)-c y_j(z).
\end{equation*}
\begin{definition}[BIC and interim IR]\label{def:bic-ir}
A mechanism is \emph{Bayesian incentive compatible} (BIC) on the full report domain if, for every buyer $i$, seller $j$, and all $v,z,c,w\in[0,1]$,
\begin{equation*}
  u_i(v;v)\ge u_i(v;z), \qquad u_j(c;c)\ge u_j(c;w).
\end{equation*}
It is \emph{interim individually rational} if
\begin{equation*}
  u_i(v;v)\ge0 \quad\text{and}\quad u_j(c;c)\ge0
\end{equation*}
for every $i,j$ and every type in $[0,1]$.
\end{definition}

Each agent may report any number in $[0,1]$, including reports that occur with probability zero under the corresponding prior.
We require the BIC and interim-IR inequalities to hold at every point in this report domain.
Appendix~\ref{app:borel-attainment} shows that the allocation and transfers can be extended to zero-probability reports without changing the second-best value.

Because agents are risk-neutral, only expected transfers conditional on reports enter the interim utilities above.
Thus, a transfer that depends on the realized matching can be replaced by its conditional expectation given the report profile without affecting BIC or interim IR.
We therefore work throughout with transfers that are deterministic conditional on reports.

We distinguish two budget notions.

\begin{definition}[Budget balance]
A mechanism is \emph{ex-ante weakly budget balanced} if
\begin{equation*}
  \sum_{i\in B}\E_{\vct t}[p_i(\vct t)] \ge \sum_{j\in S}\E_{\vct t}[r_j(\vct t)].
\end{equation*}
It is \emph{exactly strongly budget balanced at every report profile} if
\begin{equation*}
  \sum_{i\in B}p_i(\vct t) =\sum_{j\in S}r_j(\vct t) \qquad\text{for every }\vct t\in[0,1]^K.
\end{equation*}
\end{definition}

A standard reduction relates the two budget notions under independent types and interim participation constraints; see, e.g.,~\cite{BorgersNorman2009,BrustleCaiWuZhao2017}.

\begin{lemma}\label{lem:balancing}
An allocation rule $q$ admits transfers $p,r$ for which $(q,p,r)$ is BIC, interim IR, and ex-ante weakly budget balanced if and only if it admits transfers $p',r'$ for which $(q,p',r')$ is BIC, interim IR, and exactly strongly budget balanced at every report profile.
\end{lemma}

For completeness, the proof is given in Appendix~\ref{app:profile-wise-balancing}.

\subsection{First-best and second-best benchmarks}

Fix once and for all a total order on $\cF$.
At each profile $\vct t$, let $M^*(\vct t)$ be the welfare-maximizing matching selected as follows: first minimize cardinality among welfare maximizers, and then apply the fixed total order.
This tie-breaking convention is report independent.
The first-best value is
\begin{equation*}
  \FB(\mathcal I) =\E_{\vct t}\left[ \max_{M\in\cF}\sum_{(i,j)\in M}(v_i-c_j) \right] =\E_{\vct t}\left[ \sum_{(i,j)\in M^*(\vct t)}(v_i-c_j) \right].
\end{equation*}
The minimum-cardinality convention implies that every edge selected by $M^*(\vct t)$ has strictly positive surplus; a zero-surplus edge could otherwise be deleted without changing welfare.

For the classic bilateral trade problem with one buyer and one seller, Myerson and Satterthwaite showed the general impossibility of ex-post efficient trade without outside subsidies; under the standard overlapping-support conditions, the first-best trade rule cannot be implemented together with BIC, interim IR, and budget balance~\cite{MyersonSatterthwaite1983}.
This motivates the second-best benchmark: the maximum gain from trade attainable subject to the incentive, participation, and budget constraints above.
Formally, for an instance $\mathcal I$, we define the second-best value as
\begin{equation*}
  \SB(\mathcal I) :=\sup\bigl\{ \GFT(q): (q,p,r)\text{ is BIC, interim IR, and exactly strongly budget balanced} \bigr\}.
\end{equation*}

\section{The Second-Best Program and Finite-Support Reduction}
\label{sec:second-best-program}

\cite{BrustleCaiWuZhao2017} characterize ex-ante budget feasibility through monotonicity and virtual GFT.
We use the corresponding finite-support characterization under our full-report-domain and signed-transfer conventions.
Together with \cref{lem:balancing}, this allows us to formulate the second-best problem using only allocation rules and ex-ante weak budget balance.

\subsection{Reduction to finite-support priors}
\label{subsec:finite-support-reduction}

The analytic lower-bound argument will be carried out for finite-support priors.
The following proposition shows that proving $\SB\ge\beta\FB$ for every finite-support instance is enough to establish the same bound for all instances.
We use a one-sided discretization: buyer values are rounded downward, whereas seller costs are rounded upward.
The same orientation is used in the bilateral reduction of~\cite{LiuQinRenWang2026}; our proof in Appendix~\ref{app:finite-support-reduction} records the additional bookkeeping needed for multiple buyers and sellers under a general downward-closed feasible matching family.

\begin{proposition}[Finite-support reduction]
\label{prop:finite-support-reduction}
Fix \(\beta\in[0,1]\).
Suppose that
\begin{equation*}
  \SB(\mathcal J) \ge \beta\FB(\mathcal J)
\end{equation*}
for every instance \(\mathcal J\) with mutually independent finite-support priors.
Then
\begin{equation*}
  \SB(\mathcal I) \ge \beta\FB(\mathcal I)
\end{equation*}
for every instance \(\mathcal I\) with mutually independent Borel priors on \([0,1]\).
Consequently,
\begin{equation*}
  \inf_{\substack{\mathcal I\text{ with Borel priors}\\ \FB(\mathcal I)>0}} \frac{\SB(\mathcal I)}{\FB(\mathcal I)} = \inf_{\substack{\mathcal J\text{ with finite-support priors}\\ \FB(\mathcal J)>0}} \frac{\SB(\mathcal J)}{\FB(\mathcal J)}.
\end{equation*}
\end{proposition}

The proof is given in Appendix~\ref{app:finite-support-reduction}.
By \cref{prop:finite-support-reduction}, it remains to prove the half-efficiency bound for finite-support priors.

\subsection{The finite-support second-best program}
\label{subsec:finite-lagrangian}

Fix a buyer \(i\) with support $0\le v_{i,1} < v_{i,2} < \cdots <  v_{i,L_i}\le1$, define $f_{i,\ell}:= \Pr[v_i=v_{i,\ell}]>0$ and write
\begin{equation*}
  F_{i,\ell} := \sum_{h\le\ell} f_{i,h}, \qquad F_{i,0}:=0.
\end{equation*}
Define the virtual value of the buyer by
\begin{equation*}
  \phi_i(v_{i,\ell}) := v_{i,\ell} - \frac{ (v_{i,\ell+1}-v_{i,\ell})(1-F_{i,\ell}) }{ f_{i,\ell} }, \qquad \ell<L_i,
\end{equation*}
with $\phi_i(v_{i,L_i}) = v_{i,L_i}$.

For a seller \(j\) with support $0\le c_{j,1} < c_{j,2} < \cdots < c_{j,H_j} \le1$, define $g_{j,h} := \Pr[c_j=c_{j,h}] >0$ and write
\begin{equation*}
  G_{j,h} := \sum_{\ell\le h}g_{j,\ell}, \qquad G_{j,0}:=0.
\end{equation*}
Define the virtual cost of the seller by
\begin{equation*}
  \psi_j(c_{j,h}) := c_{j,h} + \frac{(c_{j,h}-c_{j,h-1})G_{j,h-1}}{g_{j,h} }, \qquad h>1,
\end{equation*}
with $\psi_j(c_{j,1}) = c_{j,1}$.

Let \(\cP\) be the finite-dimensional polytope of randomized allocation rules on the finite type space satisfying:
\begin{enumerate}[label=\textup{(\roman*)}]
  \item at every type profile, the allocation is a probability distribution over \(\cF\);
  \item every buyer's interim service is nondecreasing in her value;
  \item every seller's interim service is nonincreasing in her cost.
\end{enumerate}

For \(q\in\cP\), define its ordinary and virtual GFT by
\begin{align*}
  \cG(q) &:= \sum_{(i,j)\in E} \E_{\vct t\sim\pi} \left[ (v_i-c_j)q_{ij}(\vct t) \right],\\
  \cB(q) &:= \sum_{(i,j)\in E} \E_{\vct t\sim\pi} \left[ \bigl( \phi_i(v_i)-\psi_j(c_j) \bigr) q_{ij}(\vct t) \right].
\end{align*}
Thus, $\cB(q)$ is calculated by the same formula as ordinary GFT, except that each buyer value and seller cost is replaced by its virtual counterpart.

The following lemma is the finite-support, full-report-domain counterpart of the characterization of~\cite[Theorem~7 and Lemma~12]{BrustleCaiWuZhao2017}.

\begin{lemma}\label{lem:finite-budget-identity}
Fix $q\in\cP$.
Every full-domain BIC and interim-IR mechanism that uses $q$ at every support profile satisfies
\begin{equation*}
  \E_{\vct t\sim\pi}
  \left[
    \sum_{i\in B}p_i(\vct t)
    -
    \sum_{j\in S}r_j(\vct t)
  \right]
  \le
  \cB(q).
\end{equation*}
Equality is attainable.
Consequently, $q$ admits an ex-ante weakly budget-balanced implementation if and only if $\cB(q)\ge0$.
\end{lemma}

Equality is attained by setting the interim utility of each buyer's lowest type and each seller's highest type to zero and using the discrete BIC payment formulas to determine the corresponding interim transfers.
The proof, including the adaptation to reports outside the support, is given in Appendix~\ref{app:finite-budget-characterization}.

\begin{proposition}\label{prop:finite-primal-dual}
For finite-support priors, $\SB=\max_{q\in\cP:\,\cB(q)\ge0}\cG(q)$.
Moreover,
\begin{equation}\label{eq:lagrangian-characterization}
  \SB=\inf_{\alpha\ge0}\max_{q\in\cP}\bigl\{\cG(q)+\alpha\cB(q)\bigr\}.
\end{equation}
\end{proposition}

\begin{proof}
By \cref{lem:finite-budget-identity}, the allocation rules implementable under BIC, interim IR, and ex-ante weak budget balance are exactly the elements of $\cP$ satisfying $\cB(q)\ge0$.
By \cref{lem:balancing}, the same allocation rule also admits transfers satisfying exact profile-wise strong budget balance without changing GFT, which establishes the stated primal characterization.
The program is a finite-dimensional linear program over a nonempty compact polytope; dualizing its single inequality $\cB(q)\ge0$ with a multiplier $\alpha\ge0$ and applying strong linear-programming duality gives \eqref{eq:lagrangian-characterization}.
\end{proof}

For a fixed $\alpha$, the first term in the objective $\cG(q)+\alpha\cB(q)$ is ordinary GFT.
The second rewards allocations with positive virtual GFT and penalizes allocations with negative virtual GFT; a larger $\alpha$ places greater weight on budget feasibility.

For the rest of the finite-support proof, fix \(\alpha\ge0\).
We construct \(q^\alpha\in\cP\) such that
\begin{equation*}
  \cG(q^\alpha)+\alpha\cB(q^\alpha) \ge \frac12\FB.
\end{equation*}
The construction starts from the tie-broken first-best matching and contracts each edge-selection region induced by fixed external reports.

\section{The Second Best is Half-Efficient}
\label{sec:half-lower-bound}

The lower-bound proof in this section assumes finite-support priors; \cref{prop:finite-support-reduction} will transfer the resulting uniform guarantee to arbitrary Borel priors.
To show that the second best achieves at least a $1/2$-approximation to the first-best GFT, we first give a geometric characterization of the first best.
Our fixed-\(\alpha\) argument has three steps.
First, after fixing an edge and all external reports, we express its first-best selection region as a rank-space down-set \(\mathcal D\).
Second, we contract this region to \(\mathcal D_Z\) and prove an edge-level lower bound for the Lagrangian objective.
Third, we assemble the edge-level rules into a globally feasible and monotone allocation \(q^\alpha\).

\subsection{First-best edge-selection regions}
We begin with two structural properties of the fixed tie-broken first-best rule $M^*$: monotonicity and a geometric description of every edge-selection region.
The proofs make explicit why the matching and downward-closed assumptions are sufficient.

\begin{lemma}[Monotonicity]\label{lem:first-best-monotonicity}
    Fix all reports except one agent's report.
    \begin{enumerate}[label=\textup{(\alph*)}]
      \item If buyer $i$ is matched to seller $j$ by $M^*(v_i,\vct t_{-i})$, then she is matched to the same seller $j$ by $M^*(v_i',\vct t_{-i})$ for every $v_i'>v_i$.
      \item If seller $j$ is matched to buyer $i$ by $M^*(c_j,\vct t_{-j})$, then she is matched to the same buyer $i$ by $M^*(c_j',\vct t_{-j})$ for every $c_j'<c_j$.
    \end{enumerate}
\end{lemma}

\begin{proof}
We prove the case for the buyer, and the seller part is symmetric.
Let $M_A=M^*(v_i,\vct t_{-i})$ and suppose $(i,j)\in M_A$.
Compare $M_A$ with any other feasible matching $M_B$ at the higher report $v_i'$.

If $M_B$ leaves $i$ unmatched, its welfare is unchanged, while the welfare of $M_A$ increases by $v_i'-v_i$.
Hence $M_A$ strictly improves relative to $M_B$.
If $M_B$ also matches $i$, then both welfare values increase by $v_i'-v_i$, and the welfare comparison remains unchanged.
In the event of a welfare tie, their cardinalities and their positions in the fixed total order are also unchanged.
Thus $M_A$ continues to beat every competitor under the complete report-independent tie-breaking rule.
Therefore $M^*(v_i',\vct t_{-i})=M_A$, in particular preserving the edge $(i,j)$.
\end{proof}

Fix an edge $e=(i,j)$ and a realization $\vct t_{-ij}$ of all reports other than $(v_i,c_j)$.
For this subsection, every quantity other than $v_i$ and $c_j$ is held fixed.

\begin{lemma}\label{lem:edge-selection-wedge}
For every fixed $(e,\vct t_{-ij})$, the set of supported pairs $(v_i,c_j)$ for which $e\in M^*(v_i,c_j,\vct t_{-ij})$ is either empty or has the form
\begin{equation*}
  v_i\ge u, \qquad c_j\le d, \qquad v_i-c_j\ge t
\end{equation*}
for effective cutoffs $u,d,t$, with $t\ge0$.
Replacing $u$ by $\max\{u,t\}$ does not change the selected set.
\end{lemma}

\begin{proof}
If no feasible matching contains $e$, the edge-selection region is empty, and there is nothing to prove.
Assume henceforth that some feasible matching contains $e$.
Among matchings that contain $e$, all candidates have the same dependence $v_i-c_j$ on the two varying reports.
Let $W_e$ be the largest contribution of their remaining edges, with ties resolved by the global rule.
Thus the selected candidate containing $e$ has welfare
\begin{equation*}
  v_i-c_j+W_e.
\end{equation*}
Among matchings that do not contain $e$, separate candidates according to whether they match neither endpoint, only $i$, only $j$, or both endpoints away from each other, and let $W_{00},W_{10},W_{01},W_{11}$ be the corresponding best remaining-edge welfare values.
Their best welfare values now have the forms
\begin{equation*}
  W_{00}, \qquad v_i+W_{10}, \qquad -c_j+W_{01}, \qquad v_i-c_j+W_{11},
\end{equation*}
where an unavailable class is assigned the value $-\infty$.
Within each class, the selected candidate is independent of $(v_i,c_j)$ because all candidates in that class have the same coefficients of $v_i$ and $c_j$.

Because the priors currently have finite support, strict and weak tie-breaking inequalities can be represented by effective cutoffs on the supported type pairs.
For example, if the candidate containing $e$ must strictly beat the $W_{01}$ candidate, replace the algebraic cutoff by the smallest supported buyer value that satisfies the strict comparison.
The same convention is used for costs and attained value-cost differences.

Compare $v_i-c_j+W_e$ with the four alternatives.
The comparisons against the $01$, $10$, and $00$ classes respectively impose
\begin{equation*}
  v_i\ge W_{01}-W_e, \qquad c_j\le W_e-W_{10}, \qquad v_i-c_j\ge W_{00}-W_e,
\end{equation*}
with weak or strict boundary inclusion determined by the fixed tie-breaking rule.
On the finite support, the effective-cutoff convention converts these comparisons into the required inequalities.
The comparison against the $11$ class is independent of $(v_i,c_j)$: if the candidate containing $e$ loses it under the fixed tie-breaking rule, the edge-selection region is empty; otherwise it imposes no additional type restriction.

To prove $t\ge0$, take the selected best matching containing $e$ and delete $e$.
The resulting matching is feasible by downward closedness, matches neither $i$ nor $j$, and has remaining-edge welfare $W_e$.
Hence $W_{00}\ge W_e$.
The algebraic diagonal threshold is therefore nonnegative, and moving a strict boundary to the next attained difference preserves nonnegativity.

Finally, costs are nonnegative.
Whenever $v_i-c_j\ge t$, we have $v_i\ge t$.
Thus adding the redundant condition $v_i\ge t$, or equivalently replacing $u$ by $\max\{u,t\}$, leaves the edge-selection region unchanged.
\end{proof}

\paragraph{Rank-space representation.}
Assume henceforth that the edge-selection wedge is nonempty.
We work with the original marginal distributions \(F_i\) and \(G_j\), without conditioning on the two tail events.

For $\rho\in(0,1)$, define buyer $i$'s upper-tail quantile by
\begin{equation*}
  Q_i(\rho) := \sup\bigl\{ v\in\supp(F_i): \Pr_{v_i\sim F_i}[v_i\ge v]\ge\rho \bigr\}.
\end{equation*}
For $\sigma\in(0,1)$, define seller $j$'s lower-tail quantile by
\begin{equation*}
  Q_j(\sigma) := \inf\bigl\{ c\in\supp(G_j): \Pr_{c_j\sim G_j}[c_j\le c]\ge\sigma \bigr\}.
\end{equation*}
We extend both quantile functions to $[0,1]$ by their one-sided limits.
These endpoint values do not affect any rank-space integral but make the prefix inequalities below meaningful at $z=0$ and $z=1$.

If $\rho$ and $\sigma$ are independent $\operatorname{Unif}(0,1)$ random variables, then $Q_i(\rho)\sim F_i$ and $Q_j(\sigma)\sim G_j$.
Thus the independent type pair $(v_i,c_j)$ can equivalently be represented by the independent rank pair $(\rho,\sigma)$.

The first-best selection region of edge \(e\) in rank space is
\begin{equation*}
  \mathcal D := \bigl\{ (\rho,\sigma)\in(0,1)^2: Q_i(\rho)\ge u,\ Q_j(\sigma)\le d,\ Q_i(\rho)-Q_j(\sigma)\ge t \bigr\}.
\end{equation*}
The region \(\mathcal D\) is a down-set.
Indeed, decreasing \(\rho\) weakly increases the buyer value, while decreasing \(\sigma\) weakly decreases the seller cost, so all three inequalities defining \(\mathcal D\) are preserved.

Since $\mathcal D$ is exactly the set of rank pairs for which edge $e$ is selected by the tie-broken first-best matching, this representation gives
\begin{equation}
  \E_{\substack{v_i\sim F_i\\ c_j\sim G_j}} \left[(v_i-c_j) \mathbf 1\bigl\{e\in M^*(v_i,c_j,\vct t_{-ij}) \bigr\} \right] = \int_{\mathcal D} \bigl(Q_i(\rho)-Q_j(\sigma)\bigr) \,\dd\rho\,\dd\sigma. \label{eq:edge-first-best-rank-integral}
\end{equation}
In particular, the integral already incorporates the probability masses of the two tail restrictions \(v_i\ge u\) and \(c_j\le d\).

\subsection{An edge-level Lagrangian bound}

Continue to fix the edge $e=(i,j)$ and the external report profile $\vct t_{-ij}$.
For a fixed multiplier $\alpha\ge0$, if edge $e$ is retained at $(v_i,c_j)$, its contribution to the fixed-$\alpha$ Lagrangian objective is
\begin{equation*}
  (v_i-c_j) +\alpha\bigl(\phi_i(v_i)-\psi_j(c_j)\bigr).
\end{equation*}
The corresponding first-best contribution is the actual surplus $v_i-c_j$.
For every $\alpha>0$, we construct a random contraction of $\mathcal D$ whose expected local Lagrangian contribution captures more than one half of this first-best contribution.
The particular Beta law and exponents below are chosen so that a contraction Jacobian factor cancels one endpoint power of the Beta density, yielding the scale identity used in the one-sided bounds.

Draw
\begin{equation*}
  Z\sim\Beta\left(\frac{1}{1+\alpha},\frac{1}{1+\alpha}\right).
\end{equation*}
Here, $\Beta(a,b)$ denotes the beta distribution on $(0,1)$ with density
\begin{equation*}
  f_{a,b}(z)=\frac{z^{a-1}(1-z)^{b-1}}{\mathrm B(a,b)}, \qquad \mathrm B(a,b):=\int_0^1 z^{a-1}(1-z)^{b-1}\,\dd z.
\end{equation*}
Define
\begin{equation*}
  \theta:=Z^{\alpha/(1+\alpha)}, \qquad \eta:=(1-Z)^{\alpha/(1+\alpha)}.
\end{equation*}
As $Z$ increases, $\theta$ increases while $\eta$ decreases.
For a realization of $Z$, define
\begin{equation*}
  \mathcal D_Z =\bigl\{(\theta\rho,\eta\sigma): (\rho,\sigma)\in\mathcal D\bigr\}.
\end{equation*}
Since \(0<\theta,\eta\le1\) almost surely and \(\mathcal D\) is a down-set, the coordinatewise contraction satisfies $\mathcal D_Z\subseteq\mathcal D$.
Moreover, positive coordinatewise scaling preserves downward closedness, so \(\mathcal D_Z\) is itself a down-set.

\Cref{fig:edge-geometry} summarizes the two geometric steps behind the local construction: the first-best edge-selection wedge becomes a rank-space down-set, which is then contracted toward the lower-left corner.

\begin{figure}[t]
\centering
\begin{tikzpicture}[x=0.95cm,y=0.95cm,>=Latex,font=\normalsize]
  \begin{scope}
    \path[fill=black!10]
      (1.45,0) -- (5.05,0) -- (5.05,2.55) -- (3.35,2.55)
      -- (1.45,0.65) -- cycle;

    \draw[->,thick] (0,0) -- (5.35,0) node[right] {$v_i$};
    \draw[->,thick] (0,0) -- (0,4.05) node[above] {$c_j$};

    \draw[densely dashed] (1.45,0) -- (1.45,3.85);
    \draw[densely dashed] (0,2.55) -- (5.15,2.55);
    \draw[thick] (0.80,0) -- (4.65,3.85)
      node[pos=0.82,above,sloped,font=\small] {$v_i-c_j=t$};

    \node[below] at (1.45,0) {$u$};
    \node[left] at (0,2.55) {$d$};
    \node[align=center] at (3.55,1)
      {$e\in M^*$};

    \node[font=\normalsize] at (2.65,-0.62) {(a) Type space};
  \end{scope}

  \draw[->,thick] (5.70,1.95) -- (7.00,1.95);

  \begin{scope}[xshift=7.35cm]
    \path[fill=black!10,draw=black!70,thick]
      (0,0)
      -- (5.00,0)
      -- (5.00,0.65)
      -- (4.25,0.65)
      -- (4.25,1.25)
      -- (3.45,1.25)
      -- (3.45,1.90)
      -- (2.65,1.90)
      -- (2.65,2.60)
      -- (1.80,2.60)
      -- (1.80,3.25)
      -- (1.00,3.25)
      -- (1.00,3.70)
      -- (0,3.70)
      -- cycle;

    \begin{scope}[xscale=0.48,yscale=0.88]
      \path[fill=orange!24,fill opacity=0.65,draw=orange!80!black,very thick]
        (0,0)
        -- (5.00,0)
        -- (5.00,0.65)
        -- (4.25,0.65)
        -- (4.25,1.25)
        -- (3.45,1.25)
        -- (3.45,1.90)
        -- (2.65,1.90)
        -- (2.65,2.60)
        -- (1.80,2.60)
        -- (1.80,3.25)
        -- (1.00,3.25)
        -- (1.00,3.70)
        -- (0,3.70)
        -- cycle;
    \end{scope}

    \begin{scope}[xscale=0.88,yscale=0.48]
      \path[fill=blue!20,fill opacity=0.65,draw=blue!65!black,very thick]
        (0,0)
        -- (5.00,0)
        -- (5.00,0.65)
        -- (4.25,0.65)
        -- (4.25,1.25)
        -- (3.45,1.25)
        -- (3.45,1.90)
        -- (2.65,1.90)
        -- (2.65,2.60)
        -- (1.80,2.60)
        -- (1.80,3.25)
        -- (1.00,3.25)
        -- (1.00,3.70)
        -- (0,3.70)
        -- cycle;
    \end{scope}

    \draw[->,thick] (0,0) -- (5.35,0) node[right] {$\rho$};
    \draw[->,thick] (0,0) -- (0,4.05) node[above] {$\sigma$};

    \draw[black!70,thick] (3.15,3.16) -- (3.65,3.16);
    \node[anchor=west] at (3.75,3.16) {$\mathcal D$};
    \draw[orange!80!black,very thick] (3.15,2.72) -- (3.65,2.72);
    \node[anchor=west,orange!80!black] at (3.75,2.72) {$\mathcal D_{z_-}$};
    \draw[blue!65!black,very thick] (3.15,2.28) -- (3.65,2.28);
    \node[anchor=west,blue!65!black] at (3.75,2.28) {$\mathcal D_{z_+}$};
    \node[font=\small] at (3.10,3.92)
      {$z_-<z_+:\ \theta_-<\theta_+,\ \eta_->\eta_+$};

    \node[font=\normalsize] at (2.65,-0.62) {(b) Rank space};
  \end{scope}
\end{tikzpicture}
\caption{Schematic edge-selection regions for a fixed edge $e=(i,j)$ and fixed external reports $\vct t_{-ij}$.
Left: among supported type pairs, the tie-broken first best selects $e$ on the wedge characterized in \cref{lem:edge-selection-wedge}; the continuous shading suppresses the discreteness of the finite support.
Right: under the upper-tail buyer rank and lower-tail seller rank, the selection region becomes a down-set $\mathcal D$.
The colored regions display two realizations $z_-<z_+$, illustrating that the horizontal scale $\theta$ increases while the vertical scale $\eta$ decreases with $Z$.
Both contracted regions remain contained in $\mathcal D$.
The staircase boundary reflects the finite-support quantile representation.}
\label{fig:edge-geometry}
\end{figure}

\paragraph{Remark.}
The contracted region $\mathcal D_Z$ should not be interpreted as the trade region of a second-best mechanism.
After fixing an edge $e$, all external reports, and a multiplier $\alpha$, trading on $\mathcal D_Z$ gives the expected contribution of $e$ to the fixed-$\alpha$ Lagrangian.
Intuitively, $\mathcal D_Z$ captures how much of the first-best surplus attributable to $e$ can be retained toward a lower bound on the second-best value.

\medskip

Lastly, we show that the expected Lagrangian contribution of the contracted region $\mathcal D_Z$ is at least a constant fraction of the first-best contribution of $\mathcal D$.
Set
\begin{equation*}
  g_\alpha :=\frac{1+\alpha}{ \mathrm B(1/(1+\alpha),1/(1+\alpha))}.
\end{equation*}

\begin{lemma}
\label{lem:local-contraction}
For every $\alpha>0$, the edge-level randomized rule that retains $e$ on $\mathcal D_Z$ satisfies
\begin{equation*}
  \E_Z\left[ \int_{\mathcal D_Z} \Bigl( Q_i(\rho)-Q_j(\sigma) +\alpha\bigl( \phi_i(Q_i(\rho))-\psi_j(Q_j(\sigma)) \bigr) \Bigr) \,\dd\rho\,\dd\sigma \right] \ge g_\alpha \int_{\mathcal D} \bigl(Q_i(\rho)-Q_j(\sigma)\bigr) \,\dd\rho\,\dd\sigma.
\end{equation*}
Moreover, $g_\alpha>1/2$.
\end{lemma}

The local bound follows from separate estimates for the buyer- and seller-side terms.

\begin{lemma}[Buyer-side contraction bound]
\label{lem:buyer-contraction}
For every $\alpha>0$,
\begin{equation*}
  \E_Z\left[ \int_{\mathcal D_Z} \bigl( Q_i(\rho)+\alpha\phi_i(Q_i(\rho)) \bigr) \,\dd\rho\,\dd\sigma \right] \ge g_\alpha \int_{\mathcal D}Q_i(\rho)\,\dd\rho\,\dd\sigma.
\end{equation*}
\end{lemma}

\begin{lemma}[Seller-side contraction bound]
\label{lem:seller-contraction}
For every $\alpha>0$,
\begin{equation*}
  \E_Z\left[\int_{\mathcal D_Z} \bigl(Q_j(\sigma)+\alpha\psi_j(Q_j(\sigma))\bigr) \,\dd\rho\,\dd\sigma\right] \le g_\alpha \int_{\mathcal D}Q_j(\sigma)\,\dd\rho\,\dd\sigma.
\end{equation*}
\end{lemma}

We defer the proofs of \cref{lem:buyer-contraction,lem:seller-contraction} to \cref{subsec:one-sided-proofs}.

\begin{proof}[Proof of \cref{lem:local-contraction}]
Subtracting the seller-side upper bound in \cref{lem:seller-contraction} from the buyer-side lower bound in \cref{lem:buyer-contraction} gives the desired inequality.
It remains to bound $g_\alpha$.
Write $a:=1/(1+\alpha)$.
Splitting $1=z+(1-z)$ in the beta integral and using symmetry gives
\begin{align*}
  \mathrm B(a,a)
  &=\int_0^1 z^{a-1}(1-z)^{a-1}\bigl(z+(1-z)\bigr)\,\dd z\\
  &=2\int_0^1 z^{a-1}(1-z)^a\,\dd z
  <2\int_0^1 z^{a-1}\,\dd z
  =\frac{2}{a}=2(1+\alpha).
\end{align*}
Therefore,
\begin{equation*}
  g_\alpha = \frac{1+\alpha}{ \mathrm B(1/(1+\alpha),1/(1+\alpha)) } > \frac12. \qedhere
\end{equation*}
\end{proof}

\paragraph{Remark.}
If a buyer or seller report has positive probability, that reported type corresponds to an interval of ranks rather than a single rank.
Conditional on the reported pair, draw independent uniform hidden ranks in the corresponding buyer and seller rank intervals, apply the rank-space rule, and average over the hidden ranks.
Because the edge-level Lagrangian coefficient is constant on the rank rectangle corresponding to each reported type pair, this averaging preserves the integrals in \cref{lem:local-contraction}; because every $\mathcal D_Z$ is a down-set, the resulting report-level retention probability remains nondecreasing in the buyer's value and nonincreasing in the seller's cost.

\subsection{From edge-level bounds to half efficiency}
\label{subsec:assembly}

We now assemble the edge-level contractions into a single allocation rule for the original matching market.
For each fixed multiplier $\alpha$, the edge-level construction gives a retention probability for every edge after all external reports have been fixed.
The two points that require verification are feasibility at every support profile and monotonicity of each agent's interim service.

We first record an immediate consequence of \cref{lem:first-best-monotonicity}.

\begin{corollary}
\label{cor:stable-partner}
Fix all reports other than buyer $i$'s report.
Across all reports of buyer $i$, there is at most one seller to whom $i$ can be matched by the tie-broken first-best rule.
Symmetrically, after fixing all reports other than seller $j$'s report, there is at most one buyer to whom $j$ can be matched.
\end{corollary}

\begin{proof}
Suppose that buyer $i$ is matched at two reports $v<v'$.
Applying \cref{lem:first-best-monotonicity} at report $v$ and then increasing the report to $v'$ shows that the edge incident to $i$ is unchanged.
Hence the partner is the same at the two reports.
The seller statement is symmetric.
\end{proof}

Fix \(\alpha>0\) and a report profile \(\vct t\) in the support product.
For each edge \(e=(i,j)\in M^*(\vct t)\), apply the local Beta-contraction rule to the edge-selection region determined by the external reports \(\vct t_{-ij}\).

Define
\begin{equation*}
  \lambda_e^\alpha(\vct t)
  :=
  \Pr\!\left[e
    \text{ is retained by the local contraction rule at profile \(\vct t\)}
  \right],
\end{equation*}
where the probability is over all randomization used by the local rule.
Set $\lambda_e^\alpha(\vct t)=0$ whenever $e\notin M^*(\vct t)$.

Define a randomized allocation rule $q^\alpha$ by independently retaining each edge $e\in M^*(\vct t)$ with probability $\lambda_e^\alpha(\vct t)$.
Equivalently, for every $M\in\cF$, let
\begin{equation}
\label{eq:assembled-allocation}
  q_M^\alpha(\vct t)
  :=
  \begin{cases}
    \displaystyle
    \prod_{e\in M}\lambda_e^\alpha(\vct t)
    \prod_{e\in M^*(\vct t)\setminus M}
      \bigl(1-\lambda_e^\alpha(\vct t)\bigr),
      & M\subseteq M^*(\vct t),\\[3mm]
    0,
      & M\nsubseteq M^*(\vct t).
  \end{cases}
\end{equation}
The independent coupling in~\eqref{eq:assembled-allocation} is only for convenience, as the analysis below uses only the resulting edge marginals.

\begin{lemma}[Feasibility on the support product]
\label{lem:assembly-feasibility}
Every realization of $q^\alpha$ belongs to $\cF$.
\end{lemma}

\begin{proof}
At every support profile, each realized edge set is a subset of the single matching $M^*(\vct t)\in\cF$.
Since $\cF$ is downward closed, every such subset also belongs to $\cF$.
\end{proof}

\begin{lemma}[Monotonicity]
\label{lem:assembly-monotonicity}
For every buyer, the interim service induced by $q^\alpha$ is nondecreasing in her value.
For every seller, it is nonincreasing in her cost.
Hence $q^\alpha\in\cP$.
\end{lemma}

\begin{proof}
Fix a buyer $i$ and all other agents' reports in their supports.
By \cref{cor:stable-partner}, there is at most one seller $j^*$ to whom buyer $i$ can be matched as her report varies.
Moreover, \cref{lem:first-best-monotonicity} implies that the set of values at which $M^*$ contains $(i,j^*)$ is an upper set.

Within the edge-selection region determined by the fixed external reports, increasing the buyer value weakly decreases the buyer's upper-tail rank.
For every realization of the Beta scale, the contracted region $\mathcal D_Z$ is a down-set.
Hence its trade indicator is nondecreasing as the buyer rank improves.
Averaging over the Beta scale and over hidden ranks preserves this monotonicity.
Before the first-best threshold, the retention probability is zero, so the total service of buyer $i$ is nondecreasing over her support.
Taking expectation over the other agents' types therefore gives a nondecreasing interim service rule.

The seller argument is symmetric.
After fixing the other reports, there is at most one potential buyer partner; the first-best rule can only switch from matched to unmatched as the seller's cost increases, and membership in every contracted down-set is nonincreasing in the seller's lower-tail rank.
Thus every seller's interim service is nonincreasing in cost.
\end{proof}

Now we have all the ingredients to prove the finite-support half-efficiency bound.
\begin{lemma}
\label{lem:assembled-lagrangian}
For every $\alpha>0$,
\begin{equation}
  \label{eq:assembled-half} \cG(q^\alpha)+\alpha\cB(q^\alpha) \ge g_\alpha\FB \ge \frac12\FB,
\end{equation}
where the second inequality is strict whenever $\FB>0$.
\end{lemma}

\begin{proof}
Fix an edge $e=(i,j)$ and external reports $\vct t_{-ij}$ in their supports.
By mutual independence, after these external reports are fixed the pair $(v_i,c_j)$ still has product law $F_i\otimes G_j$.
Write
\begin{equation*}
  \ell_{ij}^\alpha(v_i,c_j) :=(v_i-c_j)+\alpha\bigl(\phi_i(v_i)-\psi_j(c_j)\bigr)
\end{equation*}
for the edge-level coefficient in the fixed-$\alpha$ Lagrangian.
If the edge-selection region determined by $\vct t_{-ij}$ is empty, both its first-best and retained contributions are zero.
Otherwise, \cref{lem:local-contraction} and \eqref{eq:edge-first-best-rank-integral} give
\begin{align}
  &\E_{\substack{v_i\sim F_i\\c_j\sim G_j}}\left[ \ell_{ij}^\alpha(v_i,c_j) \lambda_e^\alpha(v_i,c_j,\vct t_{-ij}) \right] \notag\\
  &\qquad\ge g_\alpha\, \E_{\substack{v_i\sim F_i\\c_j\sim G_j}}\left[ (v_i-c_j) \mathbf 1\{e\in M^*(v_i,c_j,\vct t_{-ij})\} \right]. \label{eq:edgewise-assembled-bound}
\end{align}
Under the original-rank convention, $\mathcal D$ includes all three wedge restrictions and already has their probability masses built into its area.
Thus no additional probability multiplier appears in~\eqref{eq:edgewise-assembled-bound}.

Average~\eqref{eq:edgewise-assembled-bound} over the external reports and sum over all edges.
By construction, the edge marginal of $q^\alpha$ is $\lambda_e^\alpha$.
The left-hand side therefore becomes
\begin{align*}
  &\sum_{(i,j)\in E} \E_{\vct t\sim\pi}\left[ \bigl( v_i-c_j +\alpha(\phi_i(v_i)-\psi_j(c_j)) \bigr) q_{ij}^\alpha(\vct t) \right]\\
  &\qquad= \cG(q^\alpha)+\alpha\cB(q^\alpha).
\end{align*}
On the right-hand side, linearity of expectation gives
\begin{align*}
  &g_\alpha \sum_{(i,j)\in E} \E_{\vct t\sim\pi}\left[ (v_i-c_j) \mathbf 1\{(i,j)\in M^*(\vct t)\} \right]\\
  &\qquad= g_\alpha \E_{\vct t\sim\pi}\left[ \sum_{(i,j)\in M^*(\vct t)}(v_i-c_j) \right] =g_\alpha\FB.
\end{align*}
There is no double counting because, at every type profile, the indicators select exactly the edges of the single tie-broken first-best matching.
This proves the first inequality in~\eqref{eq:assembled-half}.
The second follows from $g_\alpha>1/2$.
\end{proof}

\begin{theorem}[Half efficiency]
\label{thm:finite-half}
If all priors have finite support, then
\begin{equation*}
  \SB\ge\frac12\FB.
\end{equation*}
\end{theorem}

\begin{proof}
Fix $\alpha>0$.
By \cref{lem:assembly-feasibility,lem:assembly-monotonicity,lem:assembled-lagrangian},
\begin{equation*}
  \max_{q\in\cP} \bigl\{\cG(q)+\alpha\cB(q)\bigr\} \ge \cG(q^\alpha)+\alpha\cB(q^\alpha) \ge \frac12\FB.
\end{equation*}
At $\alpha=0$, the deterministic allocation that selects $M^*(\vct t)$ at every profile belongs to $\cP$ by \cref{lem:first-best-monotonicity}, and its objective value is $\FB$.
Consequently,
\begin{equation*}
  \max_{q\in\cP} \bigl\{\cG(q)+\alpha\cB(q)\bigr\} \ge\frac12\FB \qquad\text{for every }\alpha\ge0.
\end{equation*}
Taking the infimum over $\alpha$ yields
\begin{equation*}
  \SB =\inf_{\alpha\ge0} \max_{q\in\cP} \bigl\{\cG(q)+\alpha\cB(q)\bigr\} \ge\frac12\FB. \qedhere
\end{equation*}
\end{proof}

\paragraph{Remark.}
The factor \(1/2\) is tight already for bilateral trade, as shown in~\cite{LiuQinRenWang2026}.
Since bilateral trade is a special case of our downward-closed matching model, our guarantee is also tight.
\Cref{sec:subsidy} gives a fully discrete version of their instances.

\paragraph{Obtaining a half-efficient mechanism.}
The proof above establishes the half-efficiency bound but does not explicitly give a mechanism that achieves it.
Next we show how to construct a BIC, interim-IR mechanism that achieves at least half of the first-best GFT.
For $\FB>0$, compute $\FB$ and the maximum feasible matching size $R_{\max}:=\max_{M\in\cF}|M|$, and let $q^0$ be the first-best allocation.
The identities $|\mathcal D_Z|=\theta\eta|\mathcal D|$ and $\E_Z[\theta\eta]=g_\alpha/(1+\alpha)$ bound the expected number of retained edges by $g_\alpha R_{\max}/(1+\alpha)$.
Since each edge contributes at most one unit of GFT, we have $\cG(q^\alpha)\le g_\alpha R_{\max}/(1+\alpha)$.
Combining this upper bound with \cref{lem:assembled-lagrangian} and rearranging gives
\begin{equation*}
  \cB(q^\alpha)\ge\frac{g_\alpha}{\alpha(1+\alpha)}\bigl((1+\alpha)\FB-R_{\max}\bigr), \qquad \alpha>0.
\end{equation*}
Setting $\alpha=R_{\max}/\FB$ makes the term in parentheses positive:
\begin{equation*}
  (1+\alpha)\FB-R_{\max}=\left(1+\frac{R_{\max}}{\FB}\right)\FB-R_{\max}=\FB>0.
\end{equation*}
Since the coefficient $g_\alpha/(\alpha(1+\alpha))$ is also positive, $\cB(q^\alpha)>0$ at this parameter value.
For each edge and fixed external report profile, changing variables to $\mathcal D$ cancels the Beta density with the area-scaling factor, up to a continuous coefficient.
With finite supports, the remaining integrands are bounded and converge almost everywhere whenever $\alpha_n\to\alpha\ge0$.
Dominated convergence and the finite sum over edges and external support profiles therefore imply that $\cB(q^\alpha)$ is continuous, including at $\alpha=0$.
This gives us the following construction of a half-efficient mechanism: If $\cB(q^0)\ge0$, use the first-best allocation; otherwise, continuity guarantees a parameter $\widehat\alpha\in(0,R_{\max}/\FB)$ with $\cB(q^{\widehat\alpha})=0$.
At this parameter, \cref{lem:assembled-lagrangian} gives $\cG(q^{\widehat\alpha})\ge g_{\widehat\alpha}\FB>\FB/2$.
By \cref{lem:finite-budget-identity}, the selected allocation has a BIC, interim-IR, and ex-ante weakly budget-balanced implementation on the full report domain with the same GFT.
The balancing transformation in \cref{lem:balancing} then makes the transfers exactly strongly budget balanced at every report profile.
Note that this construction requires evaluating $\cB(q^\alpha)$ to choose the parameter before reports are observed.
It therefore does not assert polynomial-time computability.
If $\FB=0$, no trade and zero transfers suffice.

\subsection{Proofs of the one-sided contraction bounds}
\label{subsec:one-sided-proofs}

We now prove the two analytic estimates used in \cref{lem:local-contraction}.
Due to symmetry, we only give the buyer-side argument as a proof of \cref{lem:buyer-contraction}; the seller-side proof is analogous and can be found in Appendix~\ref{app:seller-contraction}.

\begin{proof}[Proof of \cref{lem:buyer-contraction}]
Because \(\mathcal D\) is a down-set, every horizontal section is a prefix.
Define
\begin{equation*}
  \ell_B(\sigma) := \sup\bigl\{ \rho\in[0,1]:(\rho,\sigma)\in\mathcal D \bigr\},
\end{equation*}
with the supremum of an empty section interpreted as zero.
For a fixed realization of \(Z\), write $(\rho',\sigma') := (\theta\rho,\eta\sigma)$ for the coordinates in the contracted region \(\mathcal D_Z\).
The change of variables above gives
\begin{align*}
  &\int_{\mathcal D_Z} \bigl( Q_i(\rho') +\alpha\phi_i(Q_i(\rho')) \bigr) \,\dd\rho'\,\dd\sigma'\\
  &\qquad= \theta\eta \int_0^1 \int_0^{\ell_B(\sigma)} \bigl( Q_i(\theta\rho) +\alpha\phi_i(Q_i(\theta\rho)) \bigr) \,\dd\rho\,\dd\sigma\\
  &\qquad= \eta \int_0^1 \int_0^{\theta\ell_B(\sigma)} \bigl( Q_i(r) +\alpha\phi_i(Q_i(r)) \bigr) \,\dd r\,\dd\sigma,
\end{align*}
where the second equality uses the substitution \(r=\theta\rho\).

We next establish the buyer prefix inequality
\begin{equation*}
  \int_0^z \phi_i(Q_i(r)) \,\dd r \ge zQ_i(z) \qquad \text{for every }z\in[0,1].
\end{equation*}
By the support notation introduced in \cref{subsec:finite-lagrangian},
\begin{equation*}
  Q_i(\rho)=v_{i,\ell} \qquad \text{for } \rho\in(1-F_{i,\ell},1-F_{i,\ell-1}].
\end{equation*}
For every \(\ell<L_i\), the definition of the virtual value gives
\begin{equation*}
  f_{i,\ell}\phi_i(v_{i,\ell}) = (1-F_{i,\ell-1})v_{i,\ell} -(1-F_{i,\ell})v_{i,\ell+1}.
\end{equation*}
Together with \(\phi_i(v_{i,L_i})=v_{i,L_i}\), these identities telescope to
\begin{equation*}
  \int_0^{1-F_{i,\ell-1}} \phi_i(Q_i(r))\,\dd r = (1-F_{i,\ell-1})v_{i,\ell} \qquad \text{for every }\ell.
\end{equation*}

For \(\ell<L_i\), the difference
\begin{equation*}
  \int_0^z\phi_i(Q_i(r))\,\dd r-zQ_i(z)
\end{equation*}
is affine on \((1-F_{i,\ell},1-F_{i,\ell-1}]\), has right-hand limit \((1-F_{i,\ell})(v_{i,\ell+1}-v_{i,\ell})\ge0\) at the left boundary, and equals zero at the right endpoint.
For \(\ell=L_i\), it is identically zero on \((0,1-F_{i,L_i-1}]\).
At the left boundary itself, the preceding quantile interval gives value zero.
Hence
\begin{equation*}
  \int_0^z\phi_i(Q_i(r))\,\dd r \ge zQ_i(z) \qquad \text{for every }z\in[0,1].
\end{equation*}
Consequently,
\begin{equation*}
  \int_0^z \bigl( Q_i(r)+\alpha\phi_i(Q_i(r)) \bigr) \,\dd r \ge \int_0^zQ_i(r)\,\dd r + \alpha zQ_i(z).
\end{equation*}
We next record the scale identity used on each horizontal section.
Let \(H:[0,1]\to\R\) be bounded and measurable.
Multiplying the density of \(Z\) by \(\eta=(1-Z)^{\alpha/(1+\alpha)}\) cancels the power of \(1-Z\), and hence
\begin{equation*}
  \E_Z[\eta H(\theta)] = \frac{1}{ \mathrm B(1/(1+\alpha),1/(1+\alpha)) } \int_0^1 z^{1/(1+\alpha)-1} H\bigl(z^{\alpha/(1+\alpha)}\bigr) \,\dd z.
\end{equation*}
The substitution $u=z^{\alpha/(1+\alpha)}$ gives
\begin{equation}\label{eq:scale-transform}
  \E_Z[\eta H(\theta)] = \frac{g_\alpha}{\alpha} \int_0^1 u^{1/\alpha-1}H(u)\,\dd u.
\end{equation}

Now fix a seller rank \(\sigma\), and apply \eqref{eq:scale-transform} with
\begin{equation*}
  H(u) := \int_0^{u\ell_B(\sigma)}Q_i(w)\,\dd w + \alpha u\ell_B(\sigma) Q_i\bigl(u\ell_B(\sigma)\bigr).
\end{equation*}
If \(\ell_B(\sigma)>0\), Fubini's theorem gives
\begin{equation*}
  \int_0^1 u^{1/\alpha-1} \left( \int_0^{u\ell_B(\sigma)} Q_i(w)\,\dd w \right) \,\dd u = \alpha \int_0^{\ell_B(\sigma)} \left[1- \left(\frac{w}{\ell_B(\sigma)}\right)^{1/\alpha}\right] Q_i(w)\,\dd w,
\end{equation*}
while the change of variables \(w=u\ell_B(\sigma)\) gives
\begin{equation*}
  \alpha \int_0^1 u^{1/\alpha-1} u\ell_B(\sigma) Q_i\bigl(u\ell_B(\sigma)\bigr) \,\dd u = \alpha \int_0^{\ell_B(\sigma)} \left(\frac{w}{\ell_B(\sigma)}\right)^{1/\alpha} Q_i(w)\,\dd w.
\end{equation*}
The two terms involving $\left(\frac{w}{\ell_B(\sigma)} \right)^{1/\alpha}$ cancel.
Therefore,
\begin{equation}\label{eq:scale-identity}
  \E_Z\left[ \eta \left\{ \int_0^{\theta\ell_B(\sigma)} Q_i(w)\,\dd w + \alpha\theta\ell_B(\sigma) Q_i\bigl(\theta\ell_B(\sigma)\bigr) \right\} \right] = g_\alpha \int_0^{\ell_B(\sigma)} Q_i(w)\,\dd w.
\end{equation}
The case \(\ell_B(\sigma)=0\) is immediate.
Combining the buyer prefix inequality with \eqref{eq:scale-identity} and integrating over \(\sigma\) proves the stated bound.
\end{proof}

\subsection{Applications to multi-dimensional markets}

The matching market model studied in this paper is single-dimensional, in the sense that each agent's type is a single real number that captures her value or cost for \emph{any} feasible trade.
Many applications, however, involve multi-dimensional types, where each agent's type is a vector of values or costs for different items or bundles.
In this section, we show how our half-efficiency guarantee extends to certain multi-dimensional settings.

\paragraph{Additive buyers.}
We start with the additive setting of Cai, Goldner, Ma, and Zhao~\cite{CaiGoldnerMaZhao2021} and Rubinstein, Tan, and Zhou~\cite{RubinsteinTanZhou2026}.
Let $J$ be a finite set of heterogeneous items, each initially owned by a seller.
Buyers and sellers are additive, with $v_i(T)=\sum_{j\in T}v_{ij}$ for each buyer $i$ and item set $T\subseteq J$.
Any collection of trades is feasible provided that each item is allocated to at most one buyer.
We further assume that all values and costs lie in $[0,1]$, and agents' type vectors are mutually independent; however, coordinates within one agent may be correlated.
This additive special case is discussed in both~\cite{CaiGoldnerMaZhao2021} and~\cite{RubinsteinTanZhou2026}, although the two works differ in ownership: the former has one seller per item, whereas the latter allows one seller to own all items.

Our half-efficiency theorem can be applied directly to these two additive settings, because there are no constraints linking different items, which allows us to apply our half-efficiency result separately to each item.
Each item's allocation and budget-balancing transfers depend only on reports for that item.
Independence across agents allows the single-item BIC inequalities to be summed conditional on an agent's full type, ruling out joint deviations across items even when that agent's coordinates are correlated.

\begin{corollary}[Additive multi-dimensional markets]
\label{cor:additive-multidimensional-half}
In the additive setting above, there exists a BIC, interim-IR, exactly strongly budget-balanced mechanism with
\begin{equation*}
  \GFT\ge\frac12\FB.
\end{equation*}
Moreover, the factor $1/2$ is tight.
\end{corollary}

\paragraph{A unit-demand buyer.}
\cite{CaiGoldnerMaZhao2021} further consider a single unit-demand buyer facing $n$ items, each owned by a different seller.
The buyer values item $j$ at $v_j$ and can acquire at most one item, whereas seller $j$ incurs cost $c_j$ for selling her item.
We assume that all $2n$ random variables $v_1,\ldots,v_n,c_1,\ldots,c_n$ are mutually independent, with continuous distributions supported on $[0,1]$.
Their comparison market replaces the buyer with one single-parameter buyer per item and allows at most one buyer--seller pair to trade, with each buyer interested only in her corresponding seller's item.
For $c>1$, \cite{CaiGoldnerMaZhao2021} provide a reduction stating that any guarantee of $\SB_{\rm SD}\ge\FB_{\rm SD}/c$ in this comparison market implies a $1/(2c)$ first-best guarantee in the unit-demand market~\cite[Theorem~4]{CaiGoldnerMaZhao2021}.
Our half-efficiency theorem gives $c=2$, which yields a $1/4$ guarantee in the unit-demand setting, improving the $1/6.3$ guarantee that~\cite{BabaioffRubinsteinTanWang2026} obtained by combining this reduction with their $1/3.15$ bound.

\begin{corollary}[A unit-demand buyer]
\label{cor:cai-unit-demand}
In the unit-demand setting above, there exists a DSIC, ex-post IR, ex-ante weakly budget-balanced mechanism with
\begin{equation*}
  \GFT\ge\frac14\FB.
\end{equation*}
Under the signed-transfer convention of this paper, the same allocation admits a BIC, interim-IR implementation that is exactly strongly budget balanced at every report profile.
\end{corollary}

\section{Subsidized Markets}
\label{sec:subsidy}

So far, the mechanism has been required to balance its budget.
In some markets, however, the designer can inject money into the market, for example when a platform subsidizes transactions or a public agency funds a matching program.
The impossibility theorem of Myerson and Satterthwaite concerns mechanisms without outside subsidies~\cite{MyersonSatterthwaite1983}.
If an outside party is willing to cover the deficit of the mechanism, efficient trade becomes possible.
Le, H\"am\"al\"ainen, and Barilla~\cite{LeHamalainenBarilla2026} recently characterized efficient mechanisms that are truthful in dominant strategies and minimize the worst-case deficit under interim participation.
In this section, we work under BIC and interim IR, fix the size of the subsidy, and ask what fraction of the first-best GFT it can guarantee.

The answer is determined by the coefficient $g_\alpha$ from \cref{sec:half-lower-bound}.
A subsidy only adds a linear term to the Lagrangian characterization of \cref{sec:second-best-program}, so the lower bound of \cref{sec:half-lower-bound} continues to apply.
For the matching upper bound, we use a discrete version of the bilateral instances of~\cite[Section~6]{LiuQinRenWang2026}, which show that $g_\alpha$ cannot be improved for any fixed multiplier $\alpha$.

\subsection{The subsidized second best}
\label{subsec:subsidy-model}

For $s\ge0$, we say that a mechanism $(q,p,r)$ has \emph{expected deficit at most $s$} if
\begin{equation*}
  \sum_{i\in B}\E_{\vct t}[p_i(\vct t)]-\sum_{j\in S}\E_{\vct t}[r_j(\vct t)]\ge-s.
\end{equation*}
For $s=0$, this is ex-ante weak budget balance.

\begin{definition}[Subsidized second best]
\label{def:subsidized-sb}
For an instance $\mathcal I$ and $s\ge0$, define
\begin{equation*}
  \SB_s(\mathcal I)
  :=
  \sup\bigl\{
    \GFT(q):
    (q,p,r)\text{ is BIC, interim IR, and has expected deficit at most }s
  \bigr\}.
\end{equation*}
For $\delta\ge0$, define
\begin{equation*}
  \gamma(\delta)
  :=
  \inf_{\mathcal I:\,\FB(\mathcal I)>0}
  \frac{\SB_{\delta\FB(\mathcal I)}(\mathcal I)}{\FB(\mathcal I)},
\end{equation*}
where the infimum is taken over all instances defined in \cref{sec:prelim}.
\end{definition}

By \cref{lem:balancing}, $\SB_0=\SB$.
This lemma also extends to subsidies (\cref{lem:subsidized-program}\textup{(a)} below): an allocation rule that can be implemented with expected deficit at most $s$ can also be implemented by transfers that use exactly $s$ units of outside money at every report profile.
In the definition of $\gamma$, the subsidy is measured in units of $\FB(\mathcal I)$, which the designer can compute from the priors.
Thus $\gamma(\delta)$ is the best efficiency guarantee when an outside party contributes a $\delta$ fraction of the first-best GFT.
\Cref{prop:finite-support-reduction,thm:finite-half}, together with the tight examples of~\cite{LiuQinRenWang2026}, show that $\gamma(0)=1/2$.
The following theorem determines $\gamma(\delta)$ for every $\delta\ge0$.

\begin{theorem}[Efficiency with a subsidy]
\label{thm:subsidy-frontier}
For every $\delta\ge0$,
\begin{equation}
  \label{eq:frontier}
  \gamma(\delta)
  =
  \inf_{\alpha\ge0}\bigl\{g_\alpha+\alpha\delta\bigr\},
  \qquad
  g_\alpha
  =
  \frac{1+\alpha}{\mathrm B\bigl(\frac{1}{1+\alpha},\frac{1}{1+\alpha}\bigr)}.
\end{equation}
Moreover, the infimum defining $\gamma(\delta)$ is approached by bilateral instances with finite-support priors.
\end{theorem}

Here $g_\alpha$ is the coefficient defined before \cref{lem:local-contraction}; at $\alpha=0$, the same formula gives $g_0=1$.
In \cref{sec:half-lower-bound}, we only needed the bound $g_\alpha>1/2$, which is approached as $\alpha\to\infty$.
With a positive subsidy, the infimum in~\eqref{eq:frontier} is attained at a finite multiplier, so the entire function $\alpha\mapsto g_\alpha$ becomes relevant.
We call $\gamma$ the \emph{subsidy--efficiency frontier}.
Its main properties are collected in the following corollary and illustrated in \cref{fig:frontier}.

\begin{corollary}[Shape of the frontier]
\label{cor:frontier-shape}
The function $\gamma$ is continuous, nondecreasing, and concave on $[0,\infty)$.
Moreover, the following hold.
\begin{enumerate}[label=\textup{(\alph*)}]
  \item $\gamma(0)=1/2$, and $\gamma(\delta)=1$ if and only if $\delta\ge1$.
  \item $\gamma$ is strictly increasing and strictly concave on $[0,1]$.
  For every $\delta\in(0,1)$, the infimum in~\eqref{eq:frontier} is attained at a unique $\alpha^*(\delta)\in(0,\infty)$, and $\gamma'(\delta)=\alpha^*(\delta)$.
  As $\delta$ increases from $0$ to $1$, $\alpha^*(\delta)$ decreases from $+\infty$ to $0$.
  \item As $\delta\downarrow0$ and as $\delta\uparrow1$, respectively,
  \begin{equation*}
    \gamma(\delta)
    =
    \frac12+\frac32\Bigl(\frac{\pi^2}{6}\Bigr)^{1/3}\delta^{2/3}+O(\delta),
    \qquad
    1-\gamma(\delta)
    =
    \frac{3}{2\pi^2}(1-\delta)^2+O\bigl((1-\delta)^3\bigr).
  \end{equation*}
\end{enumerate}
\end{corollary}

The proof of \cref{cor:frontier-shape} uses only standard properties of the Gamma function, and we defer it to Appendix~\ref{app:frontier-shape}.

\begin{figure}[t]
\centering
\begin{minipage}[c]{0.58\linewidth}
\centering
\begin{tikzpicture}
\begin{axis}[
  width=\linewidth,
  height=0.7\linewidth,
  xmin=0, xmax=1.2,
  ymin=0.45, ymax=1.04,
  xtick={0,0.25,0.5,0.75,1},
  xticklabels={$0$,$0.25$,$0.5$,$0.75$,$1$},
  ytick={0.5,0.6,0.7,0.8,0.9,1},
  axis lines=left,
  tick label style={font=\small},
  label style={font=\small},
  xlabel={subsidy $\delta$ (in units of $\FB$)},
  ylabel={$\gamma(\delta)$},
  clip=false,
]
  \addplot[densely dotted, black!45] coordinates {(0,0.5) (1.2,0.5)};
  \addplot[densely dotted, black!45] coordinates {(1,0.45) (1,1)};
  \addplot[very thick, blue!55!black] coordinates {(0.00000,0.50000) (0.00014,0.50447) (0.00040,0.50891) (0.00072,0.51305) (0.00118,0.51782) (0.00179,0.52316) (0.00257,0.52901) (0.00319,0.53317) (0.00390,0.53754) (0.00471,0.54208) (0.00561,0.54681) (0.00661,0.55169) (0.00772,0.55673) (0.00893,0.56192) (0.01026,0.56725) (0.01172,0.57271) (0.01329,0.57829) (0.01499,0.58399) (0.01683,0.58979) (0.01880,0.59571) (0.02091,0.60172) (0.02317,0.60782) (0.02557,0.61401) (0.02814,0.62028) (0.03086,0.62663) (0.03374,0.63305) (0.03680,0.63954) (0.04002,0.64609) (0.04343,0.65270) (0.04702,0.65936) (0.05079,0.66607) (0.05476,0.67283) (0.05893,0.67963) (0.06330,0.68647) (0.06789,0.69334) (0.07268,0.70025) (0.07770,0.70718) (0.08294,0.71413) (0.08842,0.72109) (0.09413,0.72808) (0.10009,0.73507) (0.10630,0.74207) (0.11277,0.74908) (0.11950,0.75608) (0.12649,0.76308) (0.13377,0.77007) (0.14133,0.77705) (0.14918,0.78401) (0.15733,0.79095) (0.16578,0.79786) (0.17454,0.80475) (0.18363,0.81160) (0.19304,0.81842) (0.20279,0.82519) (0.21288,0.83192) (0.22333,0.83860) (0.23414,0.84522) (0.24531,0.85178) (0.25687,0.85828) (0.26881,0.86471) (0.27493,0.86790) (0.28114,0.87107) (0.28746,0.87422) (0.29389,0.87734) (0.30041,0.88045) (0.30705,0.88354) (0.31378,0.88660) (0.32063,0.88964) (0.32758,0.89265) (0.33465,0.89565) (0.34182,0.89861) (0.34911,0.90155) (0.35652,0.90447) (0.36403,0.90735) (0.37167,0.91021) (0.37942,0.91305) (0.38729,0.91585) (0.39528,0.91862) (0.40340,0.92136) (0.41163,0.92407) (0.42000,0.92675) (0.42849,0.92939) (0.43710,0.93200) (0.44585,0.93458) (0.45472,0.93712) (0.46373,0.93962) (0.47288,0.94209) (0.48215,0.94452) (0.49157,0.94691) (0.50112,0.94926) (0.51081,0.95157) (0.52065,0.95384) (0.53062,0.95606) (0.54074,0.95825) (0.55101,0.96039) (0.56143,0.96248) (0.57199,0.96453) (0.58271,0.96654) (0.59358,0.96849) (0.60460,0.97040) (0.61579,0.97226) (0.62712,0.97407) (0.63862,0.97582) (0.65028,0.97753) (0.66211,0.97918) (0.67410,0.98077) (0.68626,0.98231) (0.69858,0.98380) (0.71108,0.98523) (0.72375,0.98660) (0.73659,0.98790) (0.74962,0.98915) (0.76282,0.99034) (0.77620,0.99146) (0.78976,0.99252) (0.80351,0.99352) (0.81745,0.99445) (0.83157,0.99531) (0.84588,0.99610) (0.86039,0.99683) (0.87509,0.99748) (0.88999,0.99806) (0.90509,0.99857) (0.92039,0.99900) (0.93590,0.99936) (0.95161,0.99964) (0.96753,0.99984) (0.98366,0.99996) (1.00000,1.00000) (1.20000,1.00000)};
  \addplot[only marks, mark=*, mark size=1.5pt, orange!80!black]
    coordinates {(0.01,0.56622) (0.1,0.73497) (0.5,0.94899)};
  \node[font=\footnotesize, anchor=west, orange!80!black] at (axis cs:0.03,0.556) {$(0.01,\,0.566)$};
  \node[font=\footnotesize, anchor=west, orange!80!black] at (axis cs:0.12,0.722) {$(0.1,\,0.735)$};
  \node[font=\footnotesize, anchor=north west, orange!80!black] at (axis cs:0.51,0.94) {$(0.5,\,0.949)$};
\end{axis}
\end{tikzpicture}
\end{minipage}\hfill
\begin{minipage}[c]{0.38\linewidth}
\centering
\small
\begin{tabular}{@{}rcc@{}}
  \toprule
  $\delta$ & $\gamma(\delta)$ & $\alpha^*(\delta)$\\
  \midrule
  $0$ & $0.5000$ & $\infty$\\
  $0.01$ & $0.5662$ & $3.924$\\
  $0.05$ & $0.6647$ & $1.755$\\
  $0.1$ & $0.7350$ & $1.151$\\
  $0.25$ & $0.8545$ & $0.565$\\
  $0.5$ & $0.9490$ & $0.243$\\
  $0.75$ & $0.9892$ & $0.093$\\
  $\ge1$ & $1$ & $0$\\
  \bottomrule
\end{tabular}
\end{minipage}
\caption{The subsidy--efficiency frontier $\gamma$ of \cref{thm:subsidy-frontier} (left), and selected values of $\gamma(\delta)$ together with the minimizing multiplier $\alpha^*(\delta)=\gamma'(\delta)$ in~\eqref{eq:frontier} (right).}
\label{fig:frontier}
\end{figure}

We highlight two consequences of \cref{cor:frontier-shape}.
First, a subsidy equal to the first-best GFT always suffices for full efficiency, while no smaller multiple of $\FB$ suffices for every instance.
The sufficiency part can also be seen from the Vickrey--Clarke--Groves (VCG) mechanism~\cite{Vickrey1961,Clarke1971,Groves1973}.
Fix a profile $\vct t$.
Since $\cF$ is downward closed, removing an agent from the market lowers the optimal welfare by at most the surplus of the edge on which she trades in $M^*(\vct t)$, and by zero if she does not trade.
Hence the VCG utility of each agent is at most the surplus of her edge, and since each edge has two endpoints, the agents' total utility is at most twice the realized surplus.
The deficit of the VCG mechanism equals this total utility minus the realized surplus, so it never exceeds the realized first-best surplus.
Taking expectations gives an expected deficit of at most $\FB(\mathcal I)$.
By \cref{lem:subsidized-program}\textup{(a)}, the same efficient allocation therefore admits a BIC and interim-IR implementation that uses exactly $\FB(\mathcal I)$ units of outside money at every report profile.
Second, even small subsidies help considerably.
For example, a subsidy of $1\%$ of the first-best GFT raises the worst-case guarantee from $50\%$ to around $56.6\%$, and a subsidy of $10\%$ raises it to around $73.5\%$.

\subsection{Proof of \texorpdfstring{\cref{thm:subsidy-frontier}}{Theorem 5.2}}
\label{subsec:frontier-proof}

The proof follows the Lagrangian approach of \cref{sec:second-best-program,sec:half-lower-bound}.
For an instance $\mathcal I$ with finite-support priors and a multiplier $\alpha\ge0$, write
\begin{equation*}
  \Lag_{\mathcal I}(\alpha)
  :=
  \max_{q\in\cP}\bigl\{\cG(q)+\alpha\cB(q)\bigr\}
\end{equation*}
for the optimal value of the fixed-$\alpha$ Lagrangian objective in \cref{prop:finite-primal-dual}.
The first lemma extends \cref{lem:balancing,prop:finite-primal-dual,prop:finite-support-reduction} to markets with a subsidy.

\begin{lemma}[Subsidized second-best program]
\label{lem:subsidized-program}
Fix $s\ge0$.
\begin{enumerate}[label=\textup{(\alph*)}]
  \item Suppose that the instance has at least two agents, as is the case whenever $\FB(\mathcal I)>0$.
  An allocation rule admits BIC and interim-IR transfers with expected deficit at most $s$ if and only if it admits BIC and interim-IR transfers satisfying
  \begin{equation*}
    \sum_{i\in B}p_i(\vct t)-\sum_{j\in S}r_j(\vct t)=-s
    \qquad\text{for every }\vct t\in[0,1]^K.
  \end{equation*}
  \item If all priors have finite support, then
  \begin{equation}
    \label{eq:subsidized-duality}
    \SB_s(\mathcal I)
    =
    \max_{q\in\cP:\,\cB(q)\ge-s}\cG(q)
    =
    \inf_{\alpha\ge0}\bigl\{\Lag_{\mathcal I}(\alpha)+\alpha s\bigr\}.
  \end{equation}
  \item For every $\delta\ge0$, the infimum defining $\gamma(\delta)$ does not change if it is restricted to instances with finite-support priors.
\end{enumerate}
\end{lemma}

The proofs are the same as those of \cref{lem:balancing,prop:finite-primal-dual,prop:finite-support-reduction}, with the budget threshold $0$ replaced by $-s$; we give the details in Appendix~\ref{app:subsidized-program}.
Note that in~\eqref{eq:subsidized-duality}, the subsidy enters only through the additive term $\alpha s$.

The second lemma states that the bound of \cref{lem:assembled-lagrangian} cannot be improved for any fixed multiplier.

\begin{lemma}[Tightness of the Beta coefficient]
\label{lem:beta-sharp}
For every $\alpha\ge0$,
\begin{equation*}
  \inf_{\mathcal I:\,\FB(\mathcal I)>0}
  \frac{\Lag_{\mathcal I}(\alpha)}{\FB(\mathcal I)}
  =
  g_\alpha,
\end{equation*}
where the infimum is taken over all instances with finite-support priors.
Moreover, the infimum is approached by bilateral instances.
\end{lemma}

The lower bound $\Lag_{\mathcal I}(\alpha)\ge g_\alpha\FB(\mathcal I)$ was established in \cref{sec:half-lower-bound}.
For $\alpha>0$, it is \cref{lem:assembled-lagrangian}; for $\alpha=0$, both sides equal $\FB(\mathcal I)$, because the first-best allocation belongs to $\cP$ by \cref{lem:first-best-monotonicity}.
We prove the matching upper bound in \cref{subsec:tight-instances}.
Assuming \cref{lem:beta-sharp} for now, we can prove \cref{thm:subsidy-frontier}.

\begin{proof}[Proof of \cref{thm:subsidy-frontier}]
Fix $\delta\ge0$.
By \cref{lem:subsidized-program}\textup{(c)}, we may restrict attention to instances with finite-support priors.
For such an instance with $\FB(\mathcal I)>0$, applying~\eqref{eq:subsidized-duality} with $s=\delta\FB(\mathcal I)$ gives
\begin{equation*}
  \frac{\SB_{\delta\FB(\mathcal I)}(\mathcal I)}{\FB(\mathcal I)}
  =
  \inf_{\alpha\ge0}
  \Bigl\{
    \frac{\Lag_{\mathcal I}(\alpha)}{\FB(\mathcal I)}+\alpha\delta
  \Bigr\}.
\end{equation*}
Taking the infimum over all such instances and interchanging the two infima, we obtain
\begin{equation*}
  \gamma(\delta)
  =
  \inf_{\alpha\ge0}
  \Bigl\{
    \inf_{\mathcal I:\,\FB(\mathcal I)>0}
    \frac{\Lag_{\mathcal I}(\alpha)}{\FB(\mathcal I)}
    +\alpha\delta
  \Bigr\}
  =
  \inf_{\alpha\ge0}\bigl\{g_\alpha+\alpha\delta\bigr\},
\end{equation*}
where the last equality is \cref{lem:beta-sharp}.
Since the inner infimum is approached by bilateral instances, so is $\gamma(\delta)$.
\end{proof}

\paragraph{Obtaining a mechanism with a subsidy.}
As in \cref{subsec:assembly}, the one-parameter family $(q^\alpha)_{\alpha\ge0}$ also yields an explicit mechanism.
Consider an instance with finite-support priors and $\FB>0$, and let $s=\delta\FB$.
If $\cB(q^0)\ge-s$, then the first-best allocation can be implemented with expected deficit at most $s$.
Otherwise, since $\alpha\mapsto\cB(q^\alpha)$ is continuous and $\cB(q^\alpha)>0$ at $\alpha=R_{\max}/\FB$, there exists $\widehat\alpha\in(0,R_{\max}/\FB)$ with $\cB(q^{\widehat\alpha})=-s$.
By \cref{lem:assembled-lagrangian},
\begin{equation*}
  \cG(q^{\widehat\alpha})
  =
  \cG(q^{\widehat\alpha})+\widehat\alpha\cB(q^{\widehat\alpha})+\widehat\alpha s
  \ge
  \bigl(g_{\widehat\alpha}+\widehat\alpha\delta\bigr)\FB
  \ge
  \gamma(\delta)\FB.
\end{equation*}
By \cref{lem:finite-budget-identity} and \cref{lem:subsidized-program}\textup{(a)}, the selected allocation is implemented by a BIC and interim-IR mechanism that uses exactly $s$ units of outside money at every report profile.
For general Borel priors, the supremum defining $\SB_s$ is attained, by the same compactness argument as in Appendix~\ref{app:borel-attainment}; see Appendix~\ref{app:subsidized-program}.

\subsection{Tight bilateral instances}
\label{subsec:tight-instances}

We now prove \cref{lem:beta-sharp} by constructing, for every fixed $\alpha>0$, bilateral instances whose Lagrangian value is close to $g_\alpha\FB$.
The idea is to make the edge-level coefficient of the Lagrangian positive at only one pair of types, while the first-best surplus is spread over the whole type space.

\begin{proof}[Proof of \cref{lem:beta-sharp}]
Fix $\alpha>0$ and let $z:=1/(1+\alpha)\in(0,1)$.
In a bilateral instance, an allocation rule is a trade probability $q(v,c)$, and, as in the proof of \cref{lem:assembled-lagrangian},
\begin{equation}
  \label{eq:bilateral-lagrangian}
  \cG(q)+\alpha\cB(q)
  =
  \E\bigl[\ell^\alpha(v,c)\,q(v,c)\bigr],
  \qquad
  \ell^\alpha(v,c)
  :=
  \bigl(v+\alpha\phi(v)\bigr)-\bigl(c+\alpha\psi(c)\bigr).
\end{equation}
We choose the buyer's value distribution so that $v+\alpha\phi(v)$ is constant below the highest value, and we let the seller's cost distribution be its mirror image.
To this end, define
\begin{equation}
  \label{eq:hard-tails}
  T_0:=1,
  \qquad
  T_{j+1}:=\frac{j+z}{j+1}\,T_j
  \quad(j\ge0),
  \qquad\text{so that}\qquad
  T_j=\frac{\Gamma(j+z)}{\Gamma(z)\Gamma(j+1)}.
\end{equation}
The sequence $(T_j)$ is positive and decreasing.
For $N\ge1$, let $\mathcal I_N$ be the bilateral instance in which the buyer's value $V$ takes values in $\{0,1/N,\ldots,1\}$ with $\Pr[V\ge j/N]=T_j$ for $0\le j\le N$, and the seller's cost is $C:=1-V'$ for an independent copy $V'$ of $V$.
In particular, $\Pr[V=1]=\Pr[C=0]=T_N$.

We first compute the coefficients $\ell^\alpha$.
Let $j<N$.
By~\eqref{eq:hard-tails}, $\Pr[V=j/N]=T_j-T_{j+1}=\frac{1-z}{j+1}T_j$ and $\Pr[V>j/N]=T_{j+1}=\frac{j+z}{j+1}T_j$.
Hence the virtual value defined in \cref{subsec:finite-lagrangian} is $\phi(j/N)=\frac jN-\frac{j+z}{(1-z)N}$, and since $\alpha=(1-z)/z$,
\begin{equation}
  \label{eq:constant-score}
  \alpha\phi(j/N)
  =
  \frac{(1-z)j-(j+z)}{zN}
  =
  -\frac{j+1}{N},
  \qquad\text{that is,}\qquad
  v+\alpha\phi(v)=-\frac1N
\end{equation}
for every $v<1$ in the support of $V$.
In fact, the recursion in~\eqref{eq:hard-tails} is exactly the condition that $v+\alpha\phi(v)$ is constant below the highest value.
At the highest value $v=1$, we have $\phi(1)=1$ and hence $v+\alpha\phi(v)=1+\alpha$.
For the seller, the definitions of $\phi$ and $\psi$ in \cref{subsec:finite-lagrangian} give $\psi(1-v)=1-\phi(v)$ for every $v$ in the support of $V$, because $C=1-V'$.
Therefore $c+\alpha\psi(c)=(1+\alpha)-\bigl((1-c)+\alpha\phi(1-c)\bigr)$, which equals $0$ at $c=0$ and $1+\alpha+1/N$ at every other cost in the support.
It follows that $\ell^\alpha(1,0)=1+\alpha$ and that $\ell^\alpha(v,c)<0$ for every other pair of types, so~\eqref{eq:bilateral-lagrangian} gives
\begin{equation}
  \label{eq:hard-lagrangian}
  \Lag_{\mathcal I_N}(\alpha)
  \le
  (1+\alpha)\Pr[V=1]\Pr[C=0]
  =
  (1+\alpha)T_N^2.
\end{equation}

We next compute the first best.
Since all types lie on the grid $\{0,1/N,\ldots,1\}$, we have $(V-C)^+=\frac1N\sum_{j=1}^N\1\{V\ge j/N>C\}$, and $\Pr[C<j/N]=\Pr[V'>1-j/N]=T_{N+1-j}$.
Hence
\begin{equation}
  \label{eq:hard-fb}
  \FB(\mathcal I_N)
  =
  \frac1N\sum_{j=1}^N T_jT_{N+1-j}.
\end{equation}
To evaluate this sum, note that $\sum_{j\ge0}T_jt^j=(1-t)^{-z}$ for $|t|<1$ by the binomial series.
Comparing the coefficients of $t^{N+1}$ in $(1-t)^{-z}(1-t)^{-z}=(1-t)^{-2z}$ and removing the terms $j=0$ and $j=N+1$ from the resulting convolution, we obtain
\begin{equation}
  \label{eq:hard-fb-closed}
  N\,\FB(\mathcal I_N)
  =
  \frac{\Gamma(N+1+2z)}{\Gamma(2z)\Gamma(N+2)}-2T_{N+1}.
\end{equation}

Finally, we let $N\to\infty$.
Since $\Gamma(n+a)/\Gamma(n+b)\sim n^{a-b}$ as $n\to\infty$, we have $T_N\sim N^{z-1}/\Gamma(z)$, and the first term on the right-hand side of~\eqref{eq:hard-fb-closed} is asymptotic to $N^{2z-1}/\Gamma(2z)$, which dominates $T_{N+1}=O(N^{z-1})$.
Combining this with~\eqref{eq:hard-lagrangian}, we conclude that
\begin{equation*}
  \limsup_{N\to\infty}\frac{\Lag_{\mathcal I_N}(\alpha)}{\FB(\mathcal I_N)}
  \le
  (1+\alpha)\frac{\Gamma(2z)}{\Gamma(z)^2}
  =
  \frac{1+\alpha}{\mathrm B(z,z)}
  =
  g_\alpha.
  \qedhere
\end{equation*}
\end{proof}

\paragraph{Remark.}
The Beta function in $g_\alpha$ has a simple origin in these instances.
Since $T_j\sim j^{z-1}/\Gamma(z)$, the normalized first best $\FB(\mathcal I_N)/T_N^2=\frac1N\sum_{j=1}^N(T_j/T_N)(T_{N+1-j}/T_N)$ is approximately a Riemann sum for $\int_0^1x^{z-1}(1-x)^{z-1}\dd x=\mathrm B(z,z)$, whereas the Lagrangian value is only $(1+\alpha)T_N^2$.
The instances $\mathcal I_N$ are a discrete counterpart of the continuous instances used in~\cite[Section~6]{LiuQinRenWang2026} to show that the factor $1/2$ is tight.
In both constructions, the buyer's value has an atom at the top and an upper tail that decays like $v^{-\alpha/(1+\alpha)}$, and for each fixed $\alpha$ the second-best ratio of their instances also converges to $g_\alpha$.
For $0<\delta<1$, choosing $\alpha=\alpha^*(\delta)$ and letting $N\to\infty$ gives instances whose subsidized second-best ratio approaches $\gamma(\delta)$.
For $\delta=0$, the half-efficiency bound is approached by first choosing $\alpha$ sufficiently large and then choosing $N$ sufficiently large.
For $\delta\ge1$, the subsidized second best equals the first best for every instance, so no limiting construction is needed.

\section{Conclusion}

We determine the exact worst-case efficiency of the second best in downward-closed matching markets: every instance satisfies $\SB\ge\FB/2$, and the constant is tight already in bilateral trade.
The lower-bound proof builds on the virtual-GFT framework of~\cite{BrustleCaiWuZhao2017} as its optimization starting point.
Its geometric step represents first-best edge-selection regions as rank-space down-sets and uses a Beta contraction to obtain an expected local Lagrangian contribution of at least half the corresponding first-best surplus before assembling the edge-level bounds globally.
This coefficient cannot be improved for any fixed Lagrange multiplier, which also determines the best efficiency guarantee for markets with an outside subsidy.

The proof also isolates two directions not addressed here.
First, the exact worst-case efficiency ratio may differ under ex-post individual rationality or when buyers can only make payments and sellers can only receive them.
Second, our finite-support mechanism is constructive for an explicitly represented instance, but obtaining a compact polynomial-time implementation under succinct representations of the feasible family or type distributions is a separate computational question.

\clearpage
\appendix
\crefalias{section}{appendix}

\section{Full-potential formulation and attainment}
\label{app:borel-attainment}

This appendix characterizes the second-best value through a full-potential program and proves that it is attained.
The formulation makes incentive compatibility on the entire report domain explicit and supplies the compactness needed for attainment.

\subsection{Reduced forms and utility potentials}
For an allocation kernel $q$, let $(x,y)$ denote its interim service rules.
Two reduced forms that agree almost everywhere under the corresponding marginal priors are identified.
Let
\begin{equation*}
  \cR\subseteq \prod_{i\in B}L^\infty(F_i) \times \prod_{j\in S}L^\infty(G_j)
\end{equation*}
be the set of all prior-almost-everywhere reduced forms induced by Borel allocation kernels that select only feasible matchings at every report profile.

An agent $k\in K$ is \emph{active} if some matching in $\cF$ contains an edge incident to $k$.
Define
\begin{equation*}
  d_k=\begin{cases}
  1,&k\text{ is active},\\
  0,&k\text{ is inactive}.
  \end{cases}
\end{equation*}
The service probability of agent $k$ always lies in $[0,d_k]$.

For any convex function $U:[0,1]\to\R$, its subdifferential relative to the interval $[0,1]$ is
\begin{equation*}
  \partial_I U(t) :=\bigl\{s\in\R: U(z)\ge U(t)+s(z-t)\text{ for every }z\in[0,1] \bigr\}.
\end{equation*}
At an endpoint, this is the appropriate one-sided relative subdifferential.

\begin{definition}[Full-potential feasible tuple]
A tuple
\begin{equation*}
  \bigl((x_i)_{i\in B},(y_j)_{j\in S},(U_k)_{k\in K}\bigr)
\end{equation*}
is \emph{full-potential feasible} if $(x,y)\in\cR$ and the following conditions hold.
\begin{enumerate}[label=\textup{(P\arabic*)}]
  \item
  For every buyer $i$, $U_i:[0,1]\to\R_{\ge0}$ is convex, nondecreasing, and $d_i$-Lipschitz, and
  \begin{equation*}
    x_i(v)\in\partial_I U_i(v)
    \qquad F_i\text{-almost everywhere}.
  \end{equation*}
  \item
  For every seller $j$, $U_j:[0,1]\to\R_{\ge0}$ is convex, nonincreasing, and $d_j$-Lipschitz, and
  \begin{equation*}
    -y_j(c)\in\partial_I U_j(c)
    \qquad G_j\text{-almost everywhere}.
  \end{equation*}
  \item
  The interim transfers implied by the potentials balance in expectation:
  \begin{equation}\label{eq:potential-balance}
    \sum_{i\in B}\E_{v_i}\bigl[v_i x_i(v_i)-U_i(v_i)\bigr] = \sum_{j\in S}\E_{c_j}\bigl[c_j y_j(c_j)+U_j(c_j)\bigr].
  \end{equation}
\end{enumerate}
\end{definition}

For a reduced form $(x,y)$, write
\begin{equation*}
  \cG(x,y) :=\sum_{i\in B}\E_{v_i}[v_i x_i(v_i)] -\sum_{j\in S}\E_{c_j}[c_j y_j(c_j)].
\end{equation*}
By iterated expectation, this is exactly the GFT of any allocation kernel inducing $(x,y)$.

\begin{theorem}\label{thm:potential-program}
For arbitrary mutually independent Borel priors on $[0,1]$,
\begin{equation}\label{eq:potential-program}
  \SB(\mathcal I) =\max\bigl\{ \cG(x,y): (x,y,U)\text{ is full-potential feasible} \bigr\}.
\end{equation}
Every maximizer induces a mechanism that is BIC and interim IR on the full report domain and has deterministic transfers satisfying exact strong budget balance at every report profile.
\end{theorem}

We prove the theorem through the following lemmas: the characterization of BIC through convex interim utilities, a completion of null reports, and compactness.

\begin{lemma}[BIC and interim IR]\label{lem:bic-subgradient}
Consider a full-domain interim service and transfer pair.
\begin{enumerate}[label=\textup{(\alph*)}]
  \item For a buyer $i$, define $U_i(v)=v x_i(v)-p_i(v)$.
  The pair is BIC and interim IR if and only if $U_i\ge0$ is convex and nondecreasing and
  \begin{equation*}
    x_i(v)\in\partial_I U_i(v)
    \qquad\text{for every }v\in[0,1].
  \end{equation*}
  \item For a seller, define $U_j(c)=r_j(c)-c y_j(c)$.
  The pair is BIC and interim IR if and only if $U_j\ge0$ is convex and nonincreasing and
  \begin{equation*}
    -y_j(c)\in\partial_I U_j(c)
    \qquad\text{for every }c\in[0,1].
  \end{equation*}
\end{enumerate}
If service is bounded by $d\in\{0,1\}$, the corresponding potential is $d$-Lipschitz.
\end{lemma}

\begin{proof}
For a buyer $i$, truthful reporting is optimal exactly when, for all $v,z\in[0,1]$, $v x_i(v)-p_i(v)\ge v x_i(z)- p_i(z)$.
Using $U_i(z)=z x_i(z)-p_i(z)$, this is equivalent to $U_i(v)\ge U_i(z)+x_i(z)(v-z)$, which says precisely that $x_i(z)\in\partial_I U_i(z)$ for every $z$.
A function admitting a subgradient at every point is convex.
Since the selected subgradients $x_i(z)$ lie in $[0,d]$, $U_i$ is nondecreasing and $d$-Lipschitz.
Interim IR naturally corresponds to $U_i\ge0$.

For a seller $j$, BIC is $U_j(c)=r_j(c)-c y_j(c)\ge r_j(z)-c y_j(z)=U_j(z)-y_j(z)(c-z)$.
Thus $-y_j(z)\in\partial_I U_j(z)$.
The remaining claims follow in the same way because the selected subgradients $-y_j(z)$ lie in $[-d,0]$.
\end{proof}

Necessity in \cref{thm:potential-program} is now immediate.
Given any feasible second-best mechanism, apply \cref{lem:bic-subgradient} to its truthful interim utilities.
The allocation kernel selects only feasible matchings at every report profile, so its reduced form belongs to $\cR$.
Taking expectations of exact profile-wise balance gives~\eqref{eq:potential-balance}.

For sufficiency, the program supplies the subgradient relation only almost everywhere.
The next lemma completes the allocation on null reports while preserving feasibility at every report profile.

\begin{lemma}[Full-domain completion]\label{lem:null-completion}
Let $(x,y)\in\cR$, and let $U$ satisfy conditions \textup{(P1)} and \textup{(P2)} of full-potential feasibility; no budget condition is required.
There exists a Borel allocation kernel $\widehat q$ that selects only feasible matchings at every report profile and full-domain service functions $\widehat x_i,\widehat y_j$ such that
\begin{align*}
  \widehat x_i(v)&\in\partial_I U_i(v) &&\text{for every buyer report }v,\\
  -\widehat y_j(c)&\in\partial_I U_j(c) &&\text{for every seller report }c,
\end{align*}
and $\widehat q$ induces these service functions at every report.
The transfers $p_i(v):=v\widehat x_i(v)-U_i(v)$ and $r_j(c):=U_j(c)+c\widehat y_j(c)$ make the resulting mechanism BIC and interim IR on the full report domain.
Moreover, $(\widehat x,\widehat y)=(x,y)$ almost everywhere, so the objective and any expected-budget identity involving $U$ are unchanged.
For finite-support priors, the discrete BIC utility functions generated by every $q\in\cP$, with each buyer's lowest type and each seller's highest type normalized to zero interim utility, satisfy these hypotheses after piecewise-linear extension.
In this case, $\widehat q$ can be chosen to agree with $q$ on every support profile, and the transfers defined above agree with the normalized discrete transfers at every support type.
\end{lemma}

\begin{proof}
Choose Borel representatives of the reduced-form coordinates.
For each buyer $i$, select a Borel function $a_i:[0,1]\to[0,d_i]$ satisfying
\begin{equation*}
  a_i(v)\in\partial_I U_i(v) \qquad\text{for every }v,
\end{equation*}
and agreeing with $x_i$ almost everywhere.
Such a selection can be obtained by keeping $x_i$ on a full-measure Borel set where the stated subgradient relation holds, and using a one-sided derivative of the convex function $U_i$ elsewhere.
One-sided derivatives of a convex function are monotone and Borel.
For a seller $j$, choose analogously $a_j(c)\in[0,d_j]$ with $-a_j(c)\in\partial_I U_j(c)$ everywhere and $a_j=y_j$ almost everywhere.

Let $q^0$ be a Borel feasible kernel inducing $(x,y)$ almost everywhere.
Choose full-measure Borel sets $H_k\subseteq[0,1]$ such that, on $H_k$, the conditional service induced by $q^0$ equals the selected function $a_k$.
Modify $q^0$ on a joint null set, if necessary, so that it is a valid probability distribution over $\cF$ at every profile.

For each active agent $k$, fix an incident edge $e_k$ that belongs to some feasible matching.
Downward closedness implies that the singleton matching $\{e_k\}$ lies in $\cF$.
Define $\widehat q$ profile by profile as follows.
\begin{enumerate}[label=\textup{(\roman*)}]
  \item If $t_k\in H_k$ for every $k$, use $q^0(\vct t)$.
  \item If exactly one coordinate, say $t_k$, lies outside $H_k$, select $\{e_k\}$ with probability $a_k(t_k)$ and select the empty matching otherwise.
For an inactive agent $a_k\equiv0$, so this rule selects the empty matching.
  \item If at least two coordinates lie outside their good sets, select the empty matching.
\end{enumerate}
The rule is Borel and selects only feasible matchings at every report profile.

Fix an arbitrary report $t_k$ of agent $k$.
Under the product prior, every other coordinate belongs to its good set almost surely.
If $t_k\in H_k$, case~(i) applies almost surely and the induced service is $a_k(t_k)$.
If $t_k\notin H_k$, case~(ii) applies almost surely and again gives service $a_k(t_k)$.
Any extra service received by the opposite endpoint of $e_k$ occurs only when $t_k\notin H_k$, a probability-zero event, and therefore does not alter that endpoint's interim service.
Hence $\widehat q$ induces the selected service functions at every report.
Since all modifications occur on null cylinders, the original reduced form, objective, and any expected-budget identity involving $U$ are preserved almost everywhere.
Define the completed interim transfers by
\begin{equation*}
  p_i(v):=v\widehat x_i(v)-U_i(v), \qquad
  r_j(c):=U_j(c)+c\widehat y_j(c).
\end{equation*}
By \cref{lem:bic-subgradient}, these transfers make the completed allocation BIC and interim IR on the full report domain.

It remains to verify the finite-support claim.
Fix a buyer with support $v_1<\cdots<v_L$ and monotone support services $x_1\le\cdots\le x_L$.
Set
\begin{equation*}
  U_1=0,\qquad
  U_\ell=\sum_{h<\ell}(v_{h+1}-v_h)x_h\quad(\ell\ge2),
\end{equation*}
and extend $U$ to $[0,1]$ with slope zero below $v_1$, slope $x_\ell$ on $[v_\ell,v_{\ell+1}]$, and slope $x_L$ above $v_L$.
The extension is nonnegative, convex, and nondecreasing, and each $x_\ell$ is a subgradient of $U$ at $v_\ell$ relative to $[0,1]$.
For a seller with support $c_1<\cdots<c_H$ and services $y_1\ge\cdots\ge y_H$, set
\begin{equation*}
  U_H=0,\qquad
  U_h=\sum_{\ell=h}^{H-1}(c_{\ell+1}-c_\ell)y_{\ell+1}\quad(h<H),
\end{equation*}
and extend $U$ with slope $-y_1$ below $c_1$, slope $-y_{h+1}$ on $[c_h,c_{h+1}]$, and slope zero above $c_H$.
This extension is nonnegative, convex, and nonincreasing, and each $-y_h$ is a subgradient of $U$ at $c_h$ relative to $[0,1]$.
Thus conditions \textup{(P1)} and \textup{(P2)} hold.

Extend the finite allocation table $q$ to a Borel kernel by selecting the empty matching outside the support product, and take each agent's good set to be her support.
The construction above then uses $q$ whenever all reports lie in their supports, so $\widehat q$ agrees with $q$ on every support profile and induces the original support services.
The completed transfer formulas above therefore agree with the normalized discrete transfers at every support type.
\end{proof}

\begin{lemma}[Full-potential implementation]\label{lem:potential-implementation}
Every full-potential feasible tuple induces a full-domain BIC and interim-IR mechanism with the same GFT and exact profile-wise strong budget balance.
\end{lemma}

\begin{proof}
If $|K|\le1$, no feasible matching contains an edge, so every reduced-form coordinate is zero.
The zero Lipschitz bound makes every potential constant, and the balance condition forces that nonnegative constant to be zero; hence the zero-allocation, zero-transfer mechanism proves the claim.
Assume henceforth that $|K|\ge2$.
Apply \cref{lem:null-completion} and write $a_i=\widehat x_i$ for buyers and $a_j=\widehat y_j$ for sellers.
Define target interim transfers by
\begin{equation*}
  p_i(v)=v a_i(v)-U_i(v), \qquad r_j(c)=U_j(c)+c a_j(c).
\end{equation*}
The full-domain subgradient inequalities and nonnegativity of the potentials imply BIC and interim IR by \cref{lem:bic-subgradient}.
Because the completion changes nothing almost everywhere, condition~\eqref{eq:potential-balance} yields
\begin{equation*}
  \sum_{i\in B}\E_{v_i}[p_i(v_i)] -\sum_{j\in S}\E_{c_j}[r_j(c_j)]=0.
\end{equation*}
Taking the preliminary transfers to depend only on each agent's own report realizes these interim transfers and gives an ex-ante exactly balanced mechanism.
Apply \cref{lem:balancing} to its interim net-payment functions $m_i=p_i$ and $m_j=-r_j$.
The resulting deterministic transfers preserve every interim payment and utility and balance exactly at every report profile.
\end{proof}

\subsection{Compactness and attainment}

\begin{lemma}[Compactness and continuity]\label{lem:potential-compactness}
The feasible set in~\eqref{eq:potential-program} is compact when reduced forms carry the product weak-star topology and potentials carry the product uniform topology.
The objective $\cG$ is continuous in this topology.
\end{lemma}

\begin{proof}
Represent an allocation kernel by the finite vector $(q_M)_{M\in\cF}$ in $L^\infty(\pi)^{|\cF|}$.
The constraints $q_M\ge0$ and $\sum_Mq_M=1$ define a weak-star closed and norm-bounded set.
By Banach--Alaoglu, the corresponding set of kernel equivalence classes is weak-star compact.
Every such class has a Borel representative that is a valid probability distribution over $\cF$ at every report profile after modification on a $\pi$-null set, so it belongs to the class of allocation kernels used in the definition of $\cR$.
Conditional expectation is weak-star continuous: for example, for every $h\in L^1(F_i)$,
\begin{equation*}
  \int h(v)x_i(v)\,\dd F_i(v) =\int h(v_i)q_i(\vct t)\,\dd\pi(\vct t),
\end{equation*}
and the right-hand side is a weak-star continuous functional of the kernel.
Therefore the reduced-form set $\cR$ is weak-star compact.

For any full-potential feasible tuple, rearranging~\eqref{eq:potential-balance} gives
\begin{equation*}
  \cG(x,y)=\sum_{k\in K}\E[U_k(t_k)].
\end{equation*}
The right-hand side is nonnegative.
Since buyer values and service probabilities are at most one and seller costs are nonnegative,
\begin{equation*}
  0\le\cG(x,y)\le |B|.
\end{equation*}
Consequently, $\E[U_k(t_k)]\le |B|$ for every $k$.
Nonnegativity and the $1$-Lipschitz bound imply, for every $t\in[0,1]$,
\begin{equation*}
  U_k(t)\le \E[U_k(t_k)]+1\le |B|+1.
\end{equation*}
Thus all feasible potentials are uniformly bounded and equicontinuous.
The Arzel\`a--Ascoli theorem makes their closure compact in the uniform topology.

It remains to verify that the subgradient coupling is closed.
Because the underlying spaces are Borel subsets of finite products of $[0,1]$, the relevant $L^1$ spaces are separable; hence the weak-star topology is metrizable on the norm-bounded subsets under consideration.
It therefore suffices to verify sequential closedness.
Consider, for a buyer, a sequence $x^n\stackrel{*}{\rightharpoonup}x$ and $U^n\to U$ uniformly with $x^n(t)\in\partial_I U^n(t)$ almost everywhere.
Fix a rational $z\in[0,1]$ and a nonnegative $h\in L^1(F_i)$.
Integrating the subgradient inequality gives
\begin{equation*}
  \int h(t)\bigl(U^n(z)-U^n(t)-x^n(t)(z-t)\bigr)\,\dd F_i(t)\ge0.
\end{equation*}
Uniform convergence handles the potential terms, and weak-star convergence handles the term containing $x^n$.
Passing to the limit shows the same integral inequality for $(x,U)$.
Since this holds for every nonnegative $h$, the integrand is nonnegative almost everywhere.
Intersecting the full-measure sets over rational $z$ and using continuity of $U$ extends the inequality to every $z\in[0,1]$.
Hence $x(t)\in\partial_I U(t)$ almost everywhere.
The seller case is identical after replacing $x$ by $-y$.

Convexity, monotonicity, Lipschitzness, and nonnegativity are preserved by uniform limits.
Equation~\eqref{eq:potential-balance} and the objective are continuous because multiplication by the bounded type coordinate is an $L^1$ test against weak-star convergence, while potential expectations are continuous under uniform convergence.
The full-potential feasible set is therefore closed inside a compact product and is itself compact.
\end{proof}

\begin{proof}[Proof of \cref{thm:potential-program}]
Necessity follows from \cref{lem:bic-subgradient} and expected profile-wise balance.
Sufficiency, including full-domain implementation and exact profile-wise balancing, is \cref{lem:potential-implementation}.
By \cref{lem:potential-compactness}, the program has a maximizer.
Thus its maximum equals the second-best value and is attained by a mechanism with all the asserted properties.
\end{proof}

\section{Proof of the budget-balance reduction}
\label{app:profile-wise-balancing}

This appendix proves \cref{lem:balancing} under the full-report-domain and signed-transfer conventions of the main text.

\begin{proof}[Proof of \cref{lem:balancing}]
If $q$ admits transfers satisfying exact strong budget balance, the same transfers satisfy ex-ante weak budget balance by taking expectations.
Conversely, suppose that $q$ admits transfers $p,r$ for which $(q,p,r)$ is BIC, interim IR, and ex-ante weakly budget balanced.
If $|K|\le1$, every feasible allocation is empty, so zero transfers give exact strong budget balance without changing $q$.
Assume henceforth that $|K|\ge2$.
For a buyer $i$, let $m_i(v)=p_i(v)$ be her interim net payment to the mechanism.
For a seller $j$, let $m_j(c)=-r_j(c)$, so a positive receipt is a negative net payment.
Thus ex-ante weak budget balance is equivalent to
\begin{equation*}
  \sum_{k\in K}\E[m_k(t_k)]\ge0.
\end{equation*}
Define
\begin{equation*}
  \sigma:=\sum_{k\in K}\E[m_k(t_k)]\ge0.
\end{equation*}
Choose an arbitrary agent $k_0$ and replace $m_{k_0}(t_{k_0})$ by $m_{k_0}(t_{k_0})-\sigma$ for every report.
This is a report-independent rebate.
It preserves BIC, weakly improves interim IR for agent $k_0$, leaves all other interim utilities unchanged, and normalizes the expected net payments so that
\begin{equation}\label{eq:mean-net-zero}
  \sum_{k\in K}\mu_k=0, \qquad \mu_k:=\E[m_k(t_k)].
\end{equation}

We now convert the interim net-payment functions $m_k(t_k)$ into net transfers $\tau_k(\vct t)$ defined at every report profile.
The construction preserves each agent's interim expected net payment while making the realized net transfers sum to zero at every report profile.
We then set $p'_i(\vct t)=\tau_i(\vct t)$ for buyers and $r'_j(\vct t)=-\tau_j(\vct t)$ for sellers.

For each report profile $\vct t\in[0,1]^K$, define
\begin{equation}\label{eq:balancing-transfer}
  \tau_k(\vct t) :=m_k(t_k)-\frac{1}{|K|-1} \sum_{\ell\ne k}\bigl(m_\ell(t_\ell)-\mu_\ell\bigr).
\end{equation}
We verify separately that this transformation preserves every interim payment and balances at every report profile.

Fix an agent $k$ and any report $t_k\in[0,1]$, including a report that has zero probability under the prior.
Because types are mutually independent, the product distribution used in the interim expectation over $\vct t_{-k}$ does not depend on this fixed report.
Hence
\begin{equation*}
  \E_{\vct t_{-k}} \bigl[m_\ell(t_\ell)-\mu_\ell\bigr]=0 \qquad\text{for every }\ell\ne k.
\end{equation*}
Taking expectations over $\vct t_{-k}$ in~\eqref{eq:balancing-transfer},
\begin{equation}\label{eq:interim-preservation}
  \E_{\vct t_{-k}}[\tau_k(t_k,\vct t_{-k})]=m_k(t_k).
\end{equation}
Hence, relative to the post-rebate interim net-payment functions, the balancing transformation preserves the full-domain interim net payment of every agent.
Since the allocation rule is unchanged, it preserves the post-rebate interim utilities and all BIC and interim-IR inequalities.

Next sum~\eqref{eq:balancing-transfer} over $k$.
Each centered term $m_\ell(t_\ell)-\mu_\ell$ appears in the inner sum for exactly $|K|-1$ agents.
Therefore, using~\eqref{eq:mean-net-zero},
\begin{align*}
  \sum_{k\in K}\tau_k(\vct t) &=\sum_{k\in K}m_k(t_k) -\sum_{\ell\in K}\bigl(m_\ell(t_\ell)-\mu_\ell\bigr)\\
  &=\sum_{\ell\in K}\mu_\ell=0
\end{align*}
for every report profile $\vct t$.
Finally, define the new buyer payment by $p_i'(\vct t)=\tau_i(\vct t)$ and the new seller receipt by $r_j'(\vct t)=-\tau_j(\vct t)$.
These transfers are deterministic conditional on reports, and the last display is exactly
\begin{equation*}
  \sum_{i\in B}p_i'(\vct t)=\sum_{j\in S}r_j'(\vct t) \qquad\text{for every }\vct t. \qedhere
\end{equation*}
\end{proof}

\begin{remark}[Where the assumptions enter]
Independence is used in~\eqref{eq:interim-preservation}; with correlated types, the centered payment of another agent need not have zero interim expectation after one agent's report is fixed.
Allowing signed transfers is important for this particular linear correction: the term in~\eqref{eq:balancing-transfer} need not respect one-sided payment constraints or ex-post individual rationality.
\end{remark}

\section{Proof of the finite-support reduction}
\label{app:finite-support-reduction}

This appendix gives the details behind \cref{prop:finite-support-reduction}.
We first bound the first-best loss from one-sided rounding and then lift a mechanism for the rounded instance to the original Borel instance.

Let
\begin{equation*}
  R_{\max}:=\max_{M\in\cF}|M|
\end{equation*}
be the maximum number of trades in a feasible matching.
For $n\ge1$, set
\begin{equation*}
  \delta_n:=2^{-n}, \qquad \Gamma_n:=\{0,\delta_n,2\delta_n,\ldots,1\}.
\end{equation*}
For $v,c\in[0,1]$, define
\begin{equation*}
  \underline v_n :=\delta_n\left\lfloor\frac{v}{\delta_n}\right\rfloor, \qquad \overline c_n :=\min\left\{1, \delta_n\left\lceil\frac{c}{\delta_n}\right\rceil \right\}.
\end{equation*}
Apply these maps coordinatewise, and let $\mathcal I_n$ be the instance with the same sets $B,S,E$ and the same feasible family $\cF$, but with buyer and seller types
\begin{equation*}
  \underline v_{i,n} \quad\text{and}\quad \overline c_{j,n},
\end{equation*}
respectively.
Since the rounding maps act coordinatewise, the rounded priors remain mutually independent and have finite support.
Write
\begin{equation*}
  \FB_n:=\FB(\mathcal I_n), \qquad \SB_n:=\SB(\mathcal I_n).
\end{equation*}
\begin{lemma}[First-best loss under one-sided rounding]
\label{lem:finite-reduction-first-best}
For every $n\ge1$,
\begin{equation}\label{eq:finite-reduction-first-best}
  \FB(\mathcal I)-2R_{\max}\delta_n \le \FB_n \le \FB(\mathcal I).
\end{equation}
In particular, $\FB_n\to\FB(\mathcal I)$.
\end{lemma}

\begin{proof}
For every edge $(i,j)$ and every type profile,
\begin{equation*}
  0 \le (v_i-c_j) -(\underline v_{i,n}-\overline c_{j,n}) \le 2\delta_n.
\end{equation*}
Hence the rounded welfare of every feasible matching is at most its welfare under the original types.
Maximizing over $\cF$ and then taking expectation gives $\FB_n\le\FB(\mathcal I)$.

Conversely, evaluate the rounded objective at the tie-broken first-best matching $M^*(\vct t)$ of the original profile.
Since $|M^*(\vct t)|\le R_{\max}$, rounding loses at most $2R_{\max}\delta_n$ on this matching.
Therefore, at every type profile,
\begin{equation*}
  \max_{M\in\cF} \sum_{(i,j)\in M} (\underline v_{i,n}-\overline c_{j,n}) \ge \sum_{(i,j)\in M^*(\vct t)}(v_i-c_j) -2R_{\max}\delta_n.
\end{equation*}
Taking expectation proves the lower bound in~\eqref{eq:finite-reduction-first-best}.
\end{proof}

\begin{lemma}[Lifting a rounded mechanism]
\label{lem:finite-reduction-lift}
Let $(q^n,p^n,r^n)$ be a full-domain BIC and interim-IR mechanism for $\mathcal I_n$ whose transfers balance exactly at every report profile.
Then $\mathcal I$ admits a full-domain BIC and interim-IR mechanism with exact profile-wise strong budget balance and GFT at least the GFT of $q^n$ in $\mathcal I_n$.
\end{lemma}

\begin{proof}
Fix $n$, write $N:=2^n$, and let
\begin{equation*}
  a_k:=k\delta_n, \qquad k=0,1,\ldots,N.
\end{equation*}
Let $x_i^n$ and $y_j^n$ be the interim service rules induced by $q^n$ under the rounded product prior, and abbreviate
\begin{equation*}
  x_{i,k}:=x_i^n(a_k), \qquad y_{j,k}:=y_j^n(a_k).
\end{equation*}
BIC implies that $x_{i,0}\le\cdots\le x_{i,N}$ for every buyer and $y_{j,0}\ge\cdots\ge y_{j,N}$ for every seller.

We first identify normalized grid transfers that weakly increase the expected difference between buyer payments and seller receipts.
Let
\begin{equation*}
  P_i^n(a_k) \quad\text{and}\quad R_j^n(a_k)
\end{equation*}
be the interim buyer payment and seller receipt in the rounded mechanism.
Define the corresponding truthful utilities
\begin{equation*}
  U_{i,k}^B :=a_kx_{i,k}-P_i^n(a_k), \qquad U_{j,k}^S :=R_j^n(a_k)-a_ky_{j,k}.
\end{equation*}
For a buyer, the adjacent BIC inequality in which type $a_{k+1}$ does not imitate $a_k$ gives
\begin{equation*}
  U_{i,k+1}^B \ge U_{i,k}^B+\delta_n x_{i,k}.
\end{equation*}
Together with $U_{i,0}^B\ge0$, this yields
\begin{equation}\label{eq:finite-reduction-buyer-rent}
  U_{i,k}^B \ge \sum_{h=0}^{k-1}\delta_n x_{i,h}.
\end{equation}
For a seller, the adjacent BIC inequality in which type $a_k$ does not imitate $a_{k+1}$ gives
\begin{equation*}
  U_{j,k}^S \ge U_{j,k+1}^S+\delta_n y_{j,k+1}.
\end{equation*}
Together with $U_{j,N}^S\ge0$, this yields
\begin{equation}\label{eq:finite-reduction-seller-rent}
  U_{j,k}^S \ge \sum_{h=k+1}^{N}\delta_n y_{j,h}.
\end{equation}
Empty sums are interpreted as zero.

Define the normalized grid transfers
\begin{align*}
  P_{i,k}^0 &:=a_kx_{i,k} -\sum_{h=0}^{k-1}\delta_n x_{i,h}, \\
  R_{j,k}^0 &:=a_ky_{j,k} +\sum_{h=k+1}^{N}\delta_n y_{j,h}.
\end{align*}
Equations~\eqref{eq:finite-reduction-buyer-rent} and~\eqref{eq:finite-reduction-seller-rent} imply
\begin{equation}\label{eq:finite-reduction-budget-comparison}
  P_{i,k}^0\ge P_i^n(a_k), \qquad R_{j,k}^0\le R_j^n(a_k).
\end{equation}
Since the original rounded mechanism balances exactly at every report profile, it balances in expectation.
Therefore~\eqref{eq:finite-reduction-budget-comparison} gives
\begin{equation}\label{eq:finite-reduction-normalized-budget}
  \sum_{i\in B}\E[P_{i,K_i}^0] -\sum_{j\in S}\E[R_{j,L_j}^0] \ge0,
\end{equation}
where $a_{K_i}=\underline v_{i,n}$ and $a_{L_j}=\overline c_{j,n}$.
The expectations in~\eqref{eq:finite-reduction-normalized-budget} are taken under the rounded priors.

We now lift the allocation to the original report domain.
Define
\begin{equation*}
  \widetilde q^{\,n}(\vct v,\vct c) :=q^n(\underline{\vct v}_n,\overline{\vct c}_n).
\end{equation*}
Feasibility is immediate: at every original report profile, the lifted rule uses a lottery over $\cF$ that is feasible at the corresponding rounded profile.
Mutual independence and the definition of the pushforward priors imply that the lifted interim services are
\begin{equation}\label{eq:finite-reduction-lifted-services}
  \widetilde x_i^{\,n}(v) =x_i^n(\underline v_n), \qquad \widetilde y_j^{\,n}(c) =y_j^n(\overline c_n).
\end{equation}
Thus every lifted buyer service is nondecreasing and every lifted seller service is nonincreasing.

Define utility potentials on the full report domain by
\begin{equation*}
  \widetilde U_i^{\,n}(v) :=\int_0^v\widetilde x_i^{\,n}(z)\,\dd z, \qquad \widetilde U_j^{\,n}(c) :=\int_c^1\widetilde y_j^{\,n}(z)\,\dd z,
\end{equation*}
and define interim transfers by
\begin{equation}\label{eq:finite-reduction-lifted-transfers}
  \widetilde P_i^{\,n}(v) :=v\widetilde x_i^{\,n}(v)-\widetilde U_i^{\,n}(v), \qquad \widetilde R_j^{\,n}(c) :=c\widetilde y_j^{\,n}(c)+\widetilde U_j^{\,n}(c).
\end{equation}
The monotonicity in~\eqref{eq:finite-reduction-lifted-services} makes the buyer potentials convex and nondecreasing and the seller potentials convex and nonincreasing.
At every jump point, the selected values $\widetilde x_i^{\,n}(v)$ and $-\widetilde y_j^{\,n}(c)$ belong to the relative subdifferentials of their respective potentials.
Hence \cref{lem:bic-subgradient} implies full-domain BIC and interim IR.

The one-sided cell conventions make the transfers in~\eqref{eq:finite-reduction-lifted-transfers} constant on rounding cells.
Specifically, for $k<N$ and $v\in[a_k,a_{k+1})$,
\begin{equation*}
  \widetilde P_i^{\,n}(v)=P_{i,k}^0,
\end{equation*}
with the same identity at $v=1$ for $k=N$.
Likewise, for $k\ge1$ and $c\in(a_{k-1},a_k]$,
\begin{equation*}
  \widetilde R_j^{\,n}(c)=R_{j,k}^0,
\end{equation*}
with the same identity at $c=0$ for $k=0$.
Consequently, taking expectation under the original priors is equivalent to taking expectation of the normalized grid transfers under the rounded priors.
By~\eqref{eq:finite-reduction-normalized-budget},
\begin{equation*}
  \sum_{i\in B}\E[\widetilde P_i^{\,n}(v_i)] -\sum_{j\in S}\E[\widetilde R_j^{\,n}(c_j)] \ge0.
\end{equation*}
Thus the lifted allocation has a full-domain BIC and interim-IR implementation that is ex-ante weakly budget balanced.
Taking the preliminary transfers to depend only on each agent's own report realizes the interim transfers above.
Applying \cref{lem:balancing} then gives deterministic transfers that balance exactly at every report profile, without changing the allocation or GFT.

It remains to compare gains from trade.
Whenever edge $(i,j)$ is selected at an original profile,
\begin{equation*}
  v_i-c_j \ge \underline v_{i,n}-\overline c_{j,n}.
\end{equation*}
Therefore, at every type profile,
\begin{equation*}
  \GFT(\widetilde q^{\,n};\vct v,\vct c) \ge \GFT(q^n;\underline{\vct v}_n,\overline{\vct c}_n).
\end{equation*}
The rounded profile on the right has exactly the product distribution of $\mathcal I_n$.
Taking expectation proves that the lifted mechanism has GFT at least the GFT of $q^n$ in the rounded instance.
\end{proof}

\begin{proof}[Proof of \cref{prop:finite-support-reduction}]
Fix an arbitrary Borel-prior instance $\mathcal I$.
If $R_{\max}=0$, then $\FB(\mathcal I)=0$ and the zero mechanism proves the claim.
Assume henceforth that $R_{\max}\ge1$.

For every $n$, the rounded instance $\mathcal I_n$ has finite-support priors.
By assumption,
\begin{equation*}
  \SB_n\ge\beta\FB_n.
\end{equation*}
Fix $\varepsilon>0$.
By the definition of $\SB_n$ as a supremum, there exists a full-domain BIC and interim-IR mechanism $(q^{n,\varepsilon},p^{n,\varepsilon},r^{n,\varepsilon})$ for $\mathcal I_n$ with exact profile-wise balance such that
\begin{equation*}
  \GFT(q^{n,\varepsilon})\ge\SB_n-\varepsilon.
\end{equation*}
Applying \cref{lem:finite-reduction-lift} to this mechanism gives a feasible mechanism for $\mathcal I$ with GFT at least $\SB_n-\varepsilon$.
Therefore,
\begin{equation*}
  \SB(\mathcal I)\ge\SB_n-\varepsilon.
\end{equation*}
Letting $\varepsilon\downarrow0$ gives $\SB(\mathcal I)\ge\SB_n$.
Consequently,
\begin{equation*}
  \SB(\mathcal I) \ge \SB_n \ge \beta\FB_n \ge \beta\bigl(\FB(\mathcal I)-2R_{\max}\delta_n\bigr),
\end{equation*}
where the last inequality is \cref{lem:finite-reduction-first-best}.
Letting $n\to\infty$ gives
\begin{equation*}
  \SB(\mathcal I)\ge\beta\FB(\mathcal I).
\end{equation*}
Every finite-support prior is a Borel prior, so the Borel worst-case ratio is at most the finite-support worst-case ratio.
Applying the implication just proved with $\beta$ equal to the finite-support infimum gives the reverse inequality.
\end{proof}

\section{Proof of the finite-support budget characterization}
\label{app:finite-budget-characterization}

This appendix proves \cref{lem:finite-budget-identity}.
The support-level payment identity is the discrete counterpart of~\cite[Theorem~7 and Lemma~12]{BrustleCaiWuZhao2017}; the final step applies our full-domain completion.

\begin{proof}[Proof of \cref{lem:finite-budget-identity}]
For one buyer, suppress the agent index and write the support as $v_1<\cdots<v_L$, probabilities $f_\ell$, and services $x_1\le\cdots\le x_L$.
BIC and interim IR imply
\begin{equation*}
  U_1\ge0,\qquad U_{\ell+1}\ge U_\ell+(v_{\ell+1}-v_\ell)x_\ell.
\end{equation*}
The smallest truthful utility vector sets $U_1=0$ and makes these upward adjacent constraints bind.
Monotonicity of the service rule ensures that the resulting utility and payment vectors satisfy all BIC constraints.
The corresponding expected payment is
\begin{align*}
  \sum_{\ell=1}^L f_\ell(v_\ell x_\ell-U_\ell)
  &=\sum_{\ell=1}^L f_\ell v_\ell x_\ell
    -\sum_{h=1}^{L-1}(v_{h+1}-v_h)x_h\sum_{\ell>h}f_\ell\\
  &=\sum_{\ell=1}^L f_\ell\phi(v_\ell)x_\ell.
\end{align*}
For one seller with support $c_1<\cdots<c_H$, probabilities $g_h$, and services $y_1\ge\cdots\ge y_H$, the smallest truthful utility vector sets $U_H=0$ and makes the downward adjacent constraints bind:
\begin{equation*}
  U_h=\sum_{\ell=h}^{H-1}(c_{\ell+1}-c_\ell)y_{\ell+1}.
\end{equation*}
Monotonicity again ensures that all BIC constraints hold.
The corresponding expected receipt is
\begin{equation*}
  \sum_{h=1}^H g_h(c_hy_h+U_h)=\sum_{h=1}^H g_h\psi(c_h)y_h.
\end{equation*}
Summing buyer payments and subtracting seller receipts, then applying iterated expectation to the interim service rules, gives exactly $\cB(q)$.
Any other BIC and interim-IR implementation weakly increases the relevant truthful utilities and therefore weakly decreases this difference.
By \cref{lem:null-completion}, the normalized allocation and transfers extend to the full report domain without changing any support-type payment or expected budget, while every full-domain mechanism restricts to the support-level constraints above.
This proves the maximal-budget claim and the final equivalence.
\end{proof}

\section{Proof of the seller contraction reduction}
\label{app:seller-contraction}
This appendix proves the seller-side contraction bound in \cref{lem:seller-contraction}.

\begin{proof}[Proof of \cref{lem:seller-contraction}]
Because \(\mathcal D\) is a down-set, every vertical section is a prefix.
Define
\begin{equation*}
  \ell_S(\rho) := \sup\bigl\{ \sigma\in[0,1]:(\rho,\sigma)\in\mathcal D \bigr\},
\end{equation*}
with the supremum of an empty section interpreted as zero.

For a fixed realization of \(Z\), write \((\rho',\sigma'):=(\theta\rho,\eta\sigma)\) for the coordinates in the contracted region \(\mathcal D_Z\).
The change of variables gives
\begin{align*}
  &\int_{\mathcal D_Z} \bigl(Q_j(\sigma')+\alpha\psi_j(Q_j(\sigma'))\bigr) \,\dd\rho'\,\dd\sigma'\\
  &\qquad= \theta\eta \int_0^1\int_0^{\ell_S(\rho)} \bigl(Q_j(\eta\sigma)+\alpha\psi_j(Q_j(\eta\sigma))\bigr) \,\dd\sigma\,\dd\rho\\
  &\qquad= \theta \int_0^1\int_0^{\eta\ell_S(\rho)} \bigl(Q_j(w)+\alpha\psi_j(Q_j(w))\bigr) \,\dd w\,\dd\rho,
\end{align*}
where the last equality uses \(w=\eta\sigma\).

We first establish the seller prefix inequality
\begin{equation*}
  \int_0^z\psi_j(Q_j(r))\,\dd r \le zQ_j(z) \qquad\text{for every }z\in[0,1].
\end{equation*}
By the support notation introduced in \cref{subsec:finite-lagrangian},
\begin{equation*}
  Q_j(\sigma)=c_{j,h} \qquad \text{for } \sigma\in(G_{j,h-1},G_{j,h}].
\end{equation*}
For the lowest seller type, $\psi_j(c_{j,1})=c_{j,1}$, so the prefix inequality holds with equality throughout $(0,G_{j,1}]$ and also at $z=0$.
For $h\ge2$, the definition of the virtual cost gives
\begin{equation*}
  g_{j,h}\psi_j(c_{j,h}) = G_{j,h}c_{j,h} -G_{j,h-1}c_{j,h-1}.
\end{equation*}
These identities telescope to
\begin{equation*}
  \int_0^{G_{j,h}}\psi_j(Q_j(r))\,\dd r =G_{j,h}c_{j,h}.
\end{equation*}
For $h\ge2$, the difference
\begin{equation*}
  zQ_j(z)-\int_0^z\psi_j(Q_j(r))\,\dd r
\end{equation*}
is affine on $(G_{j,h-1},G_{j,h}]$, has right-hand limit $G_{j,h-1}(c_{j,h}-c_{j,h-1})\ge0$ at the left boundary, and equals zero at the right endpoint.
At the left boundary itself, the preceding quantile interval gives value zero.
This proves the seller prefix inequality on all of $[0,1]$, and consequently
\begin{equation*}
  \int_0^z \bigl(Q_j(r)+\alpha\psi_j(Q_j(r))\bigr)\,\dd r \le \int_0^zQ_j(r)\,\dd r+ \alpha zQ_j(z).
\end{equation*}

Apply this inequality with $z=\eta\ell_S(\rho)$.
We use the symmetric form of the scale identity established in the proof of \cref{lem:buyer-contraction}, obtained by interchanging \(\theta\) and \(\eta\).
Thus, \eqref{eq:scale-identity} gives
\begin{equation*}
  \E_Z\left[ \theta\left\{ \int_0^{\eta\ell_S(\rho)}Q_j(w)\,\dd w +\alpha\eta\ell_S(\rho)Q_j(\eta\ell_S(\rho)) \right\} \right] =g_\alpha\int_0^{\ell_S(\rho)}Q_j(w)\,\dd w.
\end{equation*}
Integrating over $\rho$ proves the stated inequality.
\end{proof}

\section{Proof of the subsidized second-best program}
\label{app:subsidized-program}

This appendix proves \cref{lem:subsidized-program} and the attainment claim made at the end of \cref{subsec:frontier-proof}.

\begin{proof}[Proof of \cref{lem:subsidized-program}]
\emph{Part \textup{(a)}.}
Transfers whose budget equals $-s$ at every report profile have expected deficit exactly $s$.
Conversely, let $(q,p,r)$ be BIC and interim IR with expected deficit at most $s$.
As in Appendix~\ref{app:profile-wise-balancing}, write $m_i=p_i$ for buyers and $m_j=-r_j$ for sellers for the interim net payments to the mechanism, so that $\sum_{k\in K}\E[m_k(t_k)]\ge-s$.
Rebate the nonnegative constant $s+\sum_{k\in K}\E[m_k(t_k)]$ to an arbitrary agent.
This preserves BIC, weakly improves interim IR, and makes the means $\mu_k:=\E[m_k(t_k)]$ sum to $-s$.
Now define $\tau_k$ by~\eqref{eq:balancing-transfer}.
The argument following~\eqref{eq:balancing-transfer} uses only mutual independence and $|K|\ge2$.
It shows that~\eqref{eq:interim-preservation} holds at every report, so that all interim payments and utilities are preserved, and that $\sum_{k\in K}\tau_k(\vct t)=\sum_{k\in K}\mu_k=-s$ at every report profile.
Setting $p_i'=\tau_i$ for buyers and $r_j'=-\tau_j$ for sellers proves the claim.

\emph{Part \textup{(b)}.}
On a finite-support instance, BIC forces the interim services to be monotone on the supports, so every BIC mechanism uses an allocation rule in $\cP$ at the support profiles.
By \cref{lem:finite-budget-identity}, the largest expected budget of a full-domain BIC and interim-IR mechanism using $q\in\cP$ equals $\cB(q)$.
Hence the allocation rules implementable with expected deficit at most $s$ are exactly the rules $q\in\cP$ with $\cB(q)\ge-s$, and their GFT is $\cG(q)$.
These rules form a nonempty compact polytope, since the no-trade rule has $\cB=0$, so the first equality in~\eqref{eq:subsidized-duality} holds with a maximum.
Dualizing the single linear constraint $\cB(q)+s\ge0$ with a multiplier $\alpha\ge0$ and applying strong linear-programming duality gives
\begin{equation*}
  \max_{q\in\cP:\,\cB(q)\ge-s}\cG(q)
  =
  \inf_{\alpha\ge0}\max_{q\in\cP}
  \bigl\{\cG(q)+\alpha\bigl(\cB(q)+s\bigr)\bigr\}
  =
  \inf_{\alpha\ge0}\bigl\{\Lag_{\mathcal I}(\alpha)+\alpha s\bigr\}.
\end{equation*}

\emph{Part \textup{(c)}.}
Write $\gamma_{\mathrm{fin}}(\delta)$ for the infimum restricted to instances with finite-support priors.
Clearly $\gamma(\delta)\le\gamma_{\mathrm{fin}}(\delta)$.
For the reverse inequality, fix an instance $\mathcal I$ with Borel priors and $\FB(\mathcal I)>0$, and let $\mathcal I_n$ be its one-sided rounding from Appendix~\ref{app:finite-support-reduction}.
By \cref{lem:finite-reduction-first-best}, $\FB_n\le\FB(\mathcal I)$ and $\FB_n\to\FB(\mathcal I)$; in particular, $\FB_n>0$ for all large $n$.
Fix such an $n$ and $\varepsilon>0$, and choose a full-domain BIC and interim-IR mechanism for $\mathcal I_n$ with expected deficit at most $\delta\FB_n$ and GFT at least $\gamma_{\mathrm{fin}}(\delta)\FB_n-\varepsilon$.
Apply to it the lifting construction in the proof of \cref{lem:finite-reduction-lift}.
That proof uses the budget of the rounded mechanism only through its expectation.
Indeed, by~\eqref{eq:finite-reduction-budget-comparison}, passing to the normalized grid transfers can only increase the expected budget, and the lifted transfers~\eqref{eq:finite-reduction-lifted-transfers} have the same expected budget as the normalized grid transfers.
Hence the lower bound $0$ in~\eqref{eq:finite-reduction-normalized-budget} becomes $-\delta\FB_n$.
The lifted mechanism for $\mathcal I$ is thus BIC and interim IR, has expected deficit at most $\delta\FB_n\le\delta\FB(\mathcal I)$, and has GFT at least $\gamma_{\mathrm{fin}}(\delta)\FB_n-\varepsilon$.
Therefore $\SB_{\delta\FB(\mathcal I)}(\mathcal I)\ge\gamma_{\mathrm{fin}}(\delta)\FB_n-\varepsilon$.
Letting $\varepsilon\downarrow0$ and then $n\to\infty$ gives $\SB_{\delta\FB(\mathcal I)}(\mathcal I)\ge\gamma_{\mathrm{fin}}(\delta)\FB(\mathcal I)$.
\end{proof}

\paragraph{Attainment.}
We now prove that the supremum defining $\SB_s(\mathcal I)$ is attained.
Fix $s\ge0$ and an instance with Borel priors.
For a full-potential tuple, the expected budget of the transfers implied by the potentials is $\cG(x,y)-\sum_{k\in K}\E[U_k(t_k)]$.
Consider therefore the full-potential program of Appendix~\ref{app:borel-attainment} with the balance condition~\eqref{eq:potential-balance} replaced by the inequality
\begin{equation*}
  \sum_{k\in K}\E[U_k(t_k)]\le\cG(x,y)+s,
\end{equation*}
which states that the expected deficit is at most $s$.
Necessity is proved as in Appendix~\ref{app:borel-attainment}, and so is sufficiency: by \cref{lem:null-completion}, the completed transfers have expected budget $\cG(x,y)-\sum_{k}\E[U_k(t_k)]\ge-s$, and \cref{lem:subsidized-program}\textup{(a)} takes the place of \cref{lem:balancing} when $|K|\ge2$.
In the proof of \cref{lem:potential-compactness}, the uniform bound on the potentials becomes $U_k(t)\le\E[U_k(t_k)]+1\le|B|+s+1$, and the displayed inequality is preserved under the limits considered there.
Hence the modified program attains its maximum, and its value is $\SB_s(\mathcal I)$.

\section{Proof of the frontier shape}
\label{app:frontier-shape}

This appendix proves \cref{cor:frontier-shape} and derives the parametrization~\eqref{eq:frontier-param} of the frontier used for \cref{fig:frontier}.
Write $z(\alpha):=1/(1+\alpha)$, a decreasing bijection from $[0,\infty)$ onto $(0,1]$, and let
\begin{equation*}
  \chi(z):=\frac{\Gamma(1+2z)}{2\Gamma(1+z)^2},
  \qquad
  0\le z\le1.
\end{equation*}
Since $\Gamma(2z)=\Gamma(1+2z)/(2z)$ and $\Gamma(z)=\Gamma(1+z)/z$,
\begin{equation}
  \label{eq:g-chi}
  g_\alpha
  =
  \frac{1+\alpha}{\mathrm B(z,z)}
  =
  \frac1z\cdot\frac{\Gamma(2z)}{\Gamma(z)^2}
  =
  \chi\bigl(z(\alpha)\bigr),
  \qquad
  \alpha\ge0.
\end{equation}
The function $\chi$ is analytic in a neighborhood of $[0,1]$, because $\Gamma(1+2z)$ and $\Gamma(1+z)$ are analytic and nonzero for $\operatorname{Re}z>-1/2$; moreover $\chi(0)=1/2$ and $\chi(1)=1$.
Let $\Psi:=\Gamma'/\Gamma$ be the digamma function.
We use the standard facts that $\Psi(x+1)=\Psi(x)+1/x$ and hence $\Psi'(x+1)=\Psi'(x)-1/x^2$, that $\Psi(1/2)=\Psi(1)-2\log2$, and that $\Psi'(x)=\sum_{m\ge0}(x+m)^{-2}$ is strictly decreasing on $(0,\infty)$, with $\Psi'(1/2)=\pi^2/2$ and $\Psi'(1)=\pi^2/6$.

\begin{proof}[Proof of \cref{cor:frontier-shape}]
\emph{Step 1: $\chi$ is increasing and log-convex.}
Legendre's duplication formula $\Gamma(1+2z)=\pi^{-1/2}2^{2z}\Gamma(z+\tfrac12)\Gamma(z+1)$ gives
\begin{equation*}
  \log\chi(z)
  =
  (2z-1)\log2-\tfrac12\log\pi+\log\Gamma\bigl(z+\tfrac12\bigr)-\log\Gamma(z+1),
\end{equation*}
and hence
\begin{equation}
  \label{eq:logchi-derivs}
  (\log\chi)'(z)=2\log2+\Psi\bigl(z+\tfrac12\bigr)-\Psi(z+1),
  \qquad
  (\log\chi)''(z)=\Psi'\bigl(z+\tfrac12\bigr)-\Psi'(z+1)>0.
\end{equation}
Since $(\log\chi)'(0)=2\log2+\Psi(\frac12)-\Psi(1)=0$, it follows that $(\log\chi)'>0$ on $(0,1]$.
Thus $\chi'=\chi\,(\log\chi)'>0$ and $\chi''=\chi\bigl[((\log\chi)')^2+(\log\chi)''\bigr]>0$ on $(0,1]$.
At $z=1$, we have $\Psi(\frac32)-\Psi(2)=1-2\log2$ and $\Psi'(\frac32)-\Psi'(2)=(\frac{\pi^2}{2}-4)-(\frac{\pi^2}{6}-1)=\frac{\pi^2}{3}-3$, so
\begin{equation}
  \label{eq:chi-at-one}
  \chi'(1)=1,
  \qquad
  \chi''(1)=1+\Bigl(\frac{\pi^2}{3}-3\Bigr)=\frac{\pi^2}{3}-2.
\end{equation}

\emph{Step 2: $\alpha\mapsto g_\alpha$ is strictly decreasing and strictly convex.}
By~\eqref{eq:g-chi} and $z'(\alpha)=-z(\alpha)^2$,
\begin{equation}
  \label{eq:g-derivs}
  g'_\alpha=-z^2\chi'(z),
  \qquad
  g''_\alpha=2z^3\chi'(z)+z^4\chi''(z)>0,
  \qquad
  z=z(\alpha).
\end{equation}
Hence $\alpha\mapsto g'_\alpha$ is continuous and strictly increasing on $[0,\infty)$, with $g'_0=-\chi'(1)=-1$ and $g'_\alpha\to0$ as $\alpha\to\infty$, because $z^2\chi'(z)\to0$ as $z\downarrow0$.
In particular, $g'$ maps $[0,\infty)$ bijectively onto $[-1,0)$, and $g_\alpha$ decreases strictly to $\chi(0)=1/2$.

\emph{Step 3: parts \textup{(a)} and \textup{(b)}.}
As an infimum of nondecreasing affine functions of $\delta$, the frontier $\gamma$ is concave and nondecreasing on $[0,\infty)$, and hence continuous on $(0,\infty)$.
By Step~2, $\gamma(0)=\inf_\alpha g_\alpha=1/2$.
The frontier is also continuous at $0$: for every $\alpha$, $\limsup_{\delta\downarrow0}\gamma(\delta)\le g_\alpha$, and letting $\alpha\to\infty$ gives $\limsup_{\delta\downarrow0}\gamma(\delta)\le1/2=\gamma(0)\le\liminf_{\delta\downarrow0}\gamma(\delta)$.
For $\delta\ge0$, the function $\Phi_\delta(\alpha):=g_\alpha+\alpha\delta$ is strictly convex, and its derivative $g'_\alpha+\delta$ increases from $\delta-1$ at $\alpha=0$ to $\delta$ as $\alpha\to\infty$.
If $\delta\ge1$, then $\Phi_\delta$ is nondecreasing, so $\gamma(\delta)=\Phi_\delta(0)=g_0=1$.
If $\delta\in(0,1)$, then $\Phi_\delta$ has a unique critical point $\alpha^*(\delta):=(g')^{-1}(-\delta)\in(0,\infty)$, which is its unique minimizer, and $\gamma(\delta)=\Phi_\delta(\alpha^*(\delta))<\Phi_\delta(0)=1$.
This proves~(a).
Since $g''>0$, the inverse function theorem shows that $\alpha^*$ is continuously differentiable and strictly decreasing on $(0,1)$, with $\alpha^*(\delta)\to\infty$ as $\delta\downarrow0$ and $\alpha^*(\delta)\to0$ as $\delta\uparrow1$.
Differentiating $\gamma(\delta)=g_{\alpha^*(\delta)}+\alpha^*(\delta)\delta$ and using $g'_{\alpha^*(\delta)}=-\delta$ gives $\gamma'(\delta)=\alpha^*(\delta)$.
Thus $\gamma'$ is positive and strictly decreasing on $(0,1)$, so $\gamma$ is strictly increasing and strictly concave on $(0,1)$, and by continuity on $[0,1]$.
This proves~(b).

We record the first-order condition $g'_{\alpha^*}=-\delta$ in terms of $z=z(\alpha^*(\delta))\in(0,1)$.
By~\eqref{eq:g-derivs}, it reads $\delta=z^2\chi'(z)$, and then $\gamma(\delta)=\chi(z)+\alpha^*\delta=\chi(z)+\frac{1-z}{z}\,z^2\chi'(z)$, that is,
\begin{equation}
  \label{eq:frontier-param}
  \delta=z^2\chi'(z),
  \qquad
  \gamma(\delta)=\chi(z)+z(1-z)\chi'(z),
  \qquad
  \chi'(z)=2\chi(z)\bigl(\Psi(1+2z)-\Psi(1+z)\bigr).
\end{equation}
Since $z\mapsto z^2\chi'(z)=-g'_{\alpha}$ is strictly increasing in $z$ by Step~2, the first equation determines $z$ uniquely, and the second gives $\gamma(\delta)$; \cref{fig:frontier} is computed in this way.

\emph{Step 4: part \textup{(c)}.}
We first consider $\delta\uparrow1$.
Put $\varepsilon:=1-\delta$, so that $\alpha^*:=\alpha^*(\delta)\downarrow0$.
By~\eqref{eq:g-derivs} and~\eqref{eq:chi-at-one}, $g$ is smooth near $\alpha=0$, with $g_0=1$, $g'_0=-1$, and $g''_0=2\chi'(1)+\chi''(1)=\pi^2/3$.
The first-order condition $g'_{\alpha^*}-g'_0=\varepsilon$ gives $\alpha^*=\varepsilon/g''_0+O(\varepsilon^2)$.
A second-order Taylor expansion of $g$ at $0$ then yields
\begin{equation*}
  \gamma(\delta)
  =
  g_{\alpha^*}+\alpha^*(1-\varepsilon)
  =
  1-\varepsilon\alpha^*+\frac{g''_0}{2}(\alpha^*)^2+O(\varepsilon^3)
  =
  1-\frac{\varepsilon^2}{2g''_0}+O(\varepsilon^3)
  =
  1-\frac{3}{2\pi^2}\varepsilon^2+O(\varepsilon^3).
\end{equation*}

We next consider $\delta\downarrow0$, so that $z:=z(\alpha^*(\delta))\downarrow0$.
The Taylor series $\log\Gamma(1+x)=-\gamma_{\mathrm E}x+\sum_{k\ge2}(-1)^k\zeta(k)x^k/k$ for $|x|<1$, where $\gamma_{\mathrm E}$ is Euler's constant and $\zeta$ is the Riemann zeta function, gives
\begin{equation*}
  \log\bigl(2\chi(z)\bigr)
  =
  \log\Gamma(1+2z)-2\log\Gamma(1+z)
  =
  \zeta(2)z^2+O(z^3),
\end{equation*}
and hence $\chi(z)=\frac12+\frac{\zeta(2)}{2}z^2+O(z^3)$ and $\chi'(z)=\zeta(2)z+O(z^2)$.
The parametrization~\eqref{eq:frontier-param} gives $\delta=z^2\chi'(z)=\zeta(2)z^3\bigl(1+O(z)\bigr)$, so that $z^2=(\delta/\zeta(2))^{2/3}+O(\delta)$, and therefore
\begin{equation*}
  \gamma(\delta)
  =
  \chi(z)+z(1-z)\chi'(z)
  =
  \frac12+\frac32\zeta(2)z^2+O(z^3)
  =
  \frac12+\frac32\zeta(2)^{1/3}\delta^{2/3}+O(\delta).
\end{equation*}
Since $\zeta(2)=\pi^2/6$, this proves~(c).
\end{proof}

\section*{AI Disclosure}
Generative AI tools, especially GPT 5.6 and GPT 6, assisted with many parts of the mathematical development and exposition.
These tools were also used for writing the early-stage manuscript of this paper.
Despite this assistance, the authors put significant effort into verifying the results and making the text accessible to human readers, and retain full responsibility for the final text and results.
\clearpage
\phantomsection
\addcontentsline{toc}{section}{References}
\bibliographystyle{alpha}
\bibliography{secondbest}

\end{document}